%% file: Main.tex
\documentclass{LMCS}

\def\dOi{9(4:3)2013}
\lmcsheading%
{\dOi}
{1--46}
{}
{}
{Feb.~28, 2012}
{Oct.~10, 2013}
{}

\ACMCCS{[{\bf Theory of computation}]: Semantics and
  reasoning---Program Constructs---Functional constructs\,/\,Type
  structures; Semantics and reasoning---Program
  Semantics---Operational semantics\,/\,Denotational semantics}


\pdfoutput=1
\usepackage[british]{babel}
\usepackage[utf8]{inputenc}
\usepackage{hyperref}
\usepackage{microtype}
\usepackage{stmaryrd}
\usepackage{amsmath,amssymb}
\usepackage{wrapfig}

\newenvironment{definition}{\begin{defi}}{\end{defi}}
\newenvironment{theorem}{\begin{thm}}{\end{thm}}
\newenvironment{lemma}{\begin{lem}}{\end{lem}}
\newenvironment{corollary}{\begin{cor}}{\end{cor}}
\theoremstyle{definition}\newtheorem{remark}[thm]{Remark}
\theoremstyle{definition}\newtheorem{example}[thm]{Example}
\theoremstyle{definition}\newtheorem{notation}[thm]{Notation}
\theoremstyle{plain}
\theoremstyle{plain}

\usepackage{Common/appendix}
\capturecounter{thm}

\input{Common/Macros_Main}

\input{Common/Macros_Bib}
\input{Macros}

\title[Non-idempotent intersection types and strong normalisation]{Non-idempotent intersection types\\ and strong normalisation}

\author[A.~Bernadet]{Alexis Bernadet\rsuper a}
\address{{\lsuper a}{\'E}cole Polytechnique, France}
\email{Bernadet@LIX.Polytechnique.fr}

\author[S.~Graham-Lengrand]{St{\'e}phane Graham-Lengrand\rsuper b}
\address{{\lsuper b}CNRS and {\'E}cole Polytechnique, France}
\email{Graham-Lengrand@LIX.Polytechnique.fr}

\usepackage{Common/prooftree}

\begin{document}

\begin{abstract}

  We present a typing system with \emph{non-idempotent intersection} types, typing a term syntax covering three different calculi: the pure $\l$-calculus, the calculus with explicit substitutions $\LambdaS$, and the calculus with explicit substitutions, contractions and weakenings $\LambdaLxr$.

  In each of the three calculi, a term is typable if and only if it is strongly normalising, as it is the case in (many) systems with idempotent intersections.

  Non-idempotency brings extra information into typing trees, such as simple bounds on the longest reduction sequence reducing a term to its normal form.  Strong normalisation follows, without requiring reducibility techniques.

  Using this, we revisit models of the $\l$-calculus based on \emph{filters of intersection types}, and extend them to $\LambdaS$ and $\LambdaLxr$.  Non-idempotency simplifies a methodology, based on such filter models, that produces modular proofs of strong normalisation for well-known typing systems (\eg System $F$).  We also present a filter model by means of orthogonality techniques, \ie as an instance of an abstract notion of \textit{orthogonality model} formalised in this paper and inspired by classical realisability. Compared to other instances based on terms (one of which rephrases a now standard proof of strong normalisation for the $\l$-calculus), the instance based on filters is shown to be better at proving strong normalisation results for $\LambdaS$ and $\LambdaLxr$.

  Finally, the bounds on the longest reduction sequence, read off our typing trees, are refined into an \emph{exact measure}, read off a specific typing tree (called \emph{principal}); in each of the three calculi, a specific reduction sequence of such length is identified.  In the case of the $\lambda$-calculus, this complexity result is, for longest reduction sequences, the counterpart of de Carvalho's result for linear head-reduction sequences.

\end{abstract}
\maketitle

\newpage
\tableofcontents

\newpage
\section{Introduction}

\newcommand\para[1]{\par\vspace{0pt}{\bf #1.}}

\begin{wraptable}r{0pt}\vspace{-10pt}$$\infer{M \col A \cap B}{M \col A \quad M \col B}$$
\end{wraptable}
Intersection types were introduced in~\cite{CD78:newtal,CDS79:funcss}, extending the simply-typed $\lambda$-calculus with a notion of finite polymorphism.  This is achieved by a new construct $A \cap B $ in the syntax of types and new typing rules such as the one on the right, where $M\col A$ denotes that a term $M$ is of type $A$.

One of the motivations was to characterise strongly normalising (SN) $\lambda$-terms, namely the property that a $\lambda$-term can be typed if and only if it is strongly normalising.  Variants of systems using intersection types have been studied to characterise other evaluation properties of $\lambda$-terms and served as the basis of corresponding semantics~\cite{BCD:JSL83,Leivant86,kriv90,Bakel-TCS'95,Ghilezan96,AmadioCurien-book,gall98,Dezani-TCS-2004,DHM05,AlessiTCS06,CoquandSpiwack07}.

This paper develops~\cite{bernadetleng11,bernadetleng11b} with detailed proofs and extends the results to other calculi than the pure $\lambda$-calculus, namely the calculus with explicit substitutions $\LambdaS$ (a minor variant of the calculi $\lambda_s$ of~\cite{Kesner07} and $\lambda_{es}$ of~\cite{RenaudPhD}), and the calculus with explicit substitutions, explicit contractions and explicit weakenings $\LambdaLxr$~\cite{lengrandkesner,KLIaC06}:

It presents a typing system (for a syntax that covers those of $\l$, $\LambdaS$ and $\LambdaLxr$) that uses \emph{non-idempotent} intersection types (as pioneered by~\cite{WellsPOPL99,NeergaardICFP04}).

Intersections were originally introduced as idempotent, with the equation $A \cap A = A$ either as an explicit quotient or as a consequence of the system.  This corresponds to the understanding of the judgement $M \col A \cap B$ as follows: $M$ can be used as data of type $A$ or as data of type $B$.  But the meaning of $M \col A \cap B$ can be strengthened in that $M$ \textbf{will} be used \textbf{once} as data of type $A$ and \textbf{once} as data of type $B$.  With this understanding, $A \cap A \neq A$, and dropping idempotency of intersections is thus a natural way to study control of resources and complexity.

Using this typing system, the contributions of this paper are threefold:

\para{Measuring worst-case complexity}

In each of the three calculi, we refine with quantitative information the property that typability characterises strong normalisation.  Since strong normalisation ensures that all reduction sequences are finite, we are naturally interested in identifying the length of the longest reduction sequence.  This can be done as our typing system is very sensitive to the usage of resources when terms are reduced (by any of the three calculi).

Our system actually results from a long line of research inspired by Linear Logic~\cite{girard-ll}. The usual logical connectives of, say, classical and intuitionist logic, are decomposed therein into finer-grained connectives, separating a \emph{linear} part from a part that controls how and when the structural rules of \emph{contraction} and \emph{weakening} are used in proofs.  This can be seen as resource management when hypotheses, or more generally logical formulae, are considered as resource.
The Curry-Howard correspondence, which originated in the context of intuitionist logic~\cite{How:fortnc}, can be adapted to Linear Logic~\cite{Abramsky.1993,BBdPHtlca}, whose resource-awareness translates to a control of resources in the execution of programs (in the usual computational sense). From this has emerged a theory of resource $\lambda$-calculus with semantical support (such as the differential $\lambda$-calculus)~\cite{BoudolL96,BoudolCL99,EhrhardTCS03,EhrhardCoRR10,EhrhardENTCS10}. In this line of research, de~Carvalho~\cite{Carvalho05,Carvalho09corr} obtained interesting measures of reduction lengths in the $\lambda$-calculus by means of \emph{non-idempotent} intersection types: he showed a correspondence between the size of a typing derivation tree and the number of steps taken by a Krivine machine to reduce the typed $\lambda$-term, which relates to the length of linear head-reductions. But if we remain in the realm of intersection systems that characterise strong normalisation, then the more interesting measure is the length of the longest reduction sequence.

In this paper we get a result similar to de~Carvalho's, but with the measure corresponding to strong normalisation:
the length of the longest $\beta$-reduction sequence starting from any strongly normalising $\lambda$-term can be read off its typing tree in our system.\footnote{While Linear Logic also evolved typing systems that capture poly time functions~\cite{Baillot02,Baillot03,LafontTCS04,Gaboardi07}, let us emphasise that no linearity constraint is here imposed and all strongly normalising $\l$-terms can be typed (including non-linear ones). In this we also depart from the complexity results specific to the simply-typed $\lambda$-terms~\cite{Sch82,BeckmannJSL01}.}

Moreover, the idea of controlling resource usage by intersection types naturally leads to the investigation of calculi that handle resources more explicitly than the pure $\l$-calculus. While the resource calculi along the lines of~\cite{EhrhardENTCS10} are well-suited to de~Carvalho's study of head reductions, our interest in longest reduction sequences (no matter where the redexes are) lead us to explicit substitution calculi along the lines of~\cite{lengrandkesner,KLIaC06,Kesner07,RenaudPhD}. 
Hence the extension of our complexity results (already presented in~\cite{bernadetleng11} for $\l$) to $\LambdaS$ and $\LambdaLxr$.\footnote{In particular, the explicit contractions of $\LambdaLxr$ relate to the left-introduction of intersections.}

\para{Filter models and strong normalisation}

Intersection types were also used to build filter models of $\lambda$-calculus as early as~\cite{BCD:JSL83}.\footnote{For instance,~\cite{AlessiTCS06} reveals how the notion of intersection type filter can be tuned so that the corresponding filter models identify those $\lambda$-terms that are convertible by various restrictions of $\beta$- and $\eta$-conversion.} In particular,~\cite{CoquandSpiwack07} shows how filters of intersection types can be used to produce models of various type theories;~this in turn provides a modular proof that the $\lambda$-terms that are typable in some (dependent) type theory (the \emph{source system}) are typable in a unique strongly normalising system of intersection types (the \emph{target system}), and are therefore strongly normalising.

Following~\cite{bernadetleng11b}, we show here an improvement on this methodology, changing the target system of idempotent intersection types to our system of {\em non-idempotent} intersection types.\footnote{This also departs from~\cite{AlessiTCS06}.} The benefit of that move is that the strong normalisation of this new target system follows from the fact that typing trees get strictly smaller with every $\beta$-reduction. This is significantly simpler than the strong normalisation of the simply-typed $\l$-calculus and, even more so, of its extension with idempotent intersection types (for which~\cite{CoquandSpiwack07} involves reducibility techniques~\cite{Girard72,tait75}).  Strangely enough there is no price to pay for this simplification, as the construction and correctness of the filter models with respect to a source system is not made harder by non-idempotency.

While this improvement concerns any of the source systems treated in~\cite{CoquandSpiwack07}, we choose to illustrate the methodology with a concrete source system that includes the impredicative features of System~$F$~\cite{Girard72}, as suggested in the conclusion of~\cite{CoquandSpiwack07}.

Moreover, extending our improved methodology~\cite{bernadetleng11b} to the explicit substitution calculi $\LambdaS$ and $\LambdaLxr$ is a new contribution that addresses problems that are reputedly difficult: as illustrated by Melliès~\cite{Mellies:tlca1995}, strong normalisation can be hard to satisfy by an explicit substitution calculus. When it is satisfied, proof techniques often reduce the problem to the strong normalisation of pure $\l$-terms via a property known as \emph{Preservation of Strong Normalisation} (PSN)~\cite{BBLRD_JFP}, while direct proofs (\eg by reducibility~\cite{Girard72,tait75}) become hugely intricated~\cite{Dougherty-mscs-TA,LLDDvB} even in the simplest explicit substitution calculus~\lx~\cite{BlooRose:csn1995}. Here we have direct proofs of strong normalisation for $\LambdaS$ and $\LambdaLxr$, when it is typed with simple types, idempotent intersection types, System~$F$ types. These are, to our knowledge, the first direct proofs for those systems (\ie proofs that do not rely on the strong normalisation of pure $\l$-terms).

\para{Orthogonality models}

The third contribution of this paper is to show how the above methodology can be formalised in the framework of {\em orthogonality}.  Orthogonality underlies Linear Logic and its models~\cite{girard-ll} as well as classical realisability~\cite{DanosKrivine00,Krivine01,MunchCSL09}, and is used to prove properties of proofs or programs~\cite{Parigot97,mellies05recursive,LM:APAL07}.

We formalise here a parametric model construction by introducing an abstract notion of {\em orthogonality model}, which we illustrate with three different instances:
\begin{iteMize}{$\bullet$}
\item one instance is a model made of strongly normalising terms\\
(which, in the case of the pure $\l$-calculus, captures the traditional use of orthogonality to prove strong normalisation~\cite{Parigot97,LM:APAL07})
\item one instance is a model made of terms that are typable with intersection types
\item one instance is a model made of filters of intersection types
\end{iteMize}

To our knowledge, this is the first time that some filter models are shown to be captured by orthogonality techniques. Also, the systematic and modular approach offered by the abstract notion of orthogonality model facilitates the comparison of different proof techniques:
As already showed in~\cite{bernadetleng11b}, all three orthogonality models provide proofs of strong normalisation for the pure $\l$-calculus.
But here we also show that, in the case of $\LambdaS$ and $\LambdaLxr$, the term models fail to easily provide such direct proofs: one has to either infer that a term is strongly normalising from some normalisation (\resp typing) properties of its projection as a pure $\l$-term (as in the PSN property), or prove complex normalisation (\resp typing) properties within $\LambdaS$ and $\LambdaLxr$ themselves (as in the {\bf IE} property identified in~\cite{Kesner09lmcs}).
On the contrary, the filter model provides strong normalisation results for $\LambdaS$ and $\LambdaLxr$ as smoothly as for the pure $\lambda$-calculus.

\para{Structure of the paper}

This paper aims at factorising as much material as possible between the three calculi, and present the material specific to each of them in a systematic way.

Section~\ref{sec:calculus} presents the generic syntax that covers those of $\l$-calculus, $\LambdaS$ and $\LambdaLxr$; it presents the (non-idempotent) intersection types, the typing system used in the rest of this paper and its basic properties.

Section~\ref{sec:soundness} proves Subject Reduction for each of the three calculi, showing that typing trees get smaller with every reduction, from which strong normalisation is inferred (\emph{Soundness}).

Section~\ref{sec:semantics} presents the filter structure of our intersection types, and the construction of a filter model for a very general \emph{source typing system}, which is thus proved strongly normalising in each of the three calculi; the abstract notion of \emph{orthogonality model} is defined with sufficient conditions for the Adequacy Lemma to hold (being typed implies having a semantics in the model); three instances of orthogonality models are defined and compared in the view of proving strong normalisation results for the three calculi.

Section~\ref{sec:completeness} proves Subject Expansion for each of the three calculi, from which typing derivations are shown to exist for every strongly normalising terms (\emph{Completeness}); such derivations are proved to satisfy some specific properties called \emph{optimality} and \emph{principality}.

Section~\ref{sec:complexity} draws advantage of the optimality and principality properties to refine the upper bound on longest reduction sequences into an exact measure that is reached by some specific reduction sequence; this is done in each of the three calculi.

Section~\ref{sec:other} discusses alternative measures for the explicit substitution calculi $\LambdaS$ and $\LambdaLxr$, and Section~\ref{sec:conclusion} concludes. An appendix details the proofs of the theorems that would otherwise overload the paper with technicalities.

\section{The calculus}

\label{sec:calculus}

The intersection type system we define here was first designed for the pure $\l$-calculus.  However, it can easily be extended to other calculi such as the explicit substitution calculus $\LambdaS$, or the explicit substitution calculus $\LambdaLxr$ where weakenings and contractions are also explicit.

The theories of those three calculi share a lot of material, which is why we present them in parallel, factorising what can be factorised: For instance, the syntaxes of the three calculi are fragments of a common grammar, for which a generic intersection type system specifies a notion of typing for each fragment.  However, the calculi do not share the same reduction rules.

In this section, we first present the common grammar, then we define our generic intersection type system for it.

\subsection{Terms}

The syntaxes of the three calculi are subsets of a common grammar defined
as follows:
\[M, N \recdef x \mid \l x . M \mid M N \mid \Subst M x N \mid
\weakening x M \mid \contraction x y z M\]
The free variables $fv(M)$ of a term $M$ are defined by the rules of figure~
\ref{fig:freevars}.

\begin{figure}[!h]
  \[\begin{array}{|c|}
    \upline
    fv(x) = \{ x \} \qquad fv(\l x . M) = fv(M) - \{ x \} \qquad
    fv(M N) = fv(M) \cup fv(N)\\\\
    fv(\Subst M x N) = (fv(M) - \{ x \}) \cup fv(N) \qquad
    fv(\weakening x M) = fv(M) \cup \{ x \}\\\\
    fv(\contraction x y z M) = \begin{cases}
      (fv(M) - \{y, z\}) \cup \{ x \} \text{ if } y \in fv(M) \text{ or }
      z \in fv(M)\\
      fv(M) \text{ otherwise}
    \end{cases}
    \downline
  \end{array}\]
  \caption{Free variables of a term}
  \label{fig:freevars}
\end{figure}

We consider terms up to $\alpha$-equivalence and use Barendregt's convention~\cite{Bar84} to avoid variable capture.

\begin{definition}[Linear terms]\strut

  \begin{iteMize}{$\bullet$}

  \item $x$ is linear.

  \item If $M$ and $N$ are linear and $fv(M) \cap fv(N) = \varnothing$, then
    $M N$ is linear.

  \item If $M$ is linear and $x \in fv(M)$, then $\l x . M$ is linear.

  \item If $M$ and $N$ are linear, $x \in fv(M)$ and
    $(fv(M) - \{x \}) \cap fv(N) = \varnothing$, then $\Subst M x N$ is linear.

  \item If $M$ is linear and $x \notin fv(M)$, then $\weakening x M$ is linear.

  \item If $M$ is linear, $y \in fv(M)$, $z \in fv(M)$, and $x \notin fv(M)$, then
    $\contraction x y z M$ is linear.

  \end{iteMize}
\end{definition}

In this paper we consider in particular the three following fragments of the above syntax.

\begin{definition}[Fragments]\strut

  \begin{enumerate}[\hbox to8 pt{\hfill}]  

  \item\noindent{\hskip-12 pt\bf Pure $\l$-calculus:}\ A $\l$-term is a term $M$ that does not contain
    $\Subst M x N$, or $\weakening x M$ or $\contraction x y z M$.

  \item\noindent{\hskip-12 pt\bf $\LambdaS$-calculus:}\ A $\LambdaS$-term is a term that does not contain
    $\weakening x M$ or $\contraction x y z M$.

  \item\noindent{\hskip-12 pt\bf $\LambdaLxr$-calculus:}\ A $\LambdaLxr$-term is a term $M$ that is linear:
    every free and bound variable appears once and only once in the term
    (see~\cite{KLIaC06}).

  \end{enumerate}

\end{definition}

\subsection{Types}

To define the intersection type system, we first define the intersection types,
which are the same for the three calculi.

\begin{definition}[Types]\strut

  We consider a countable infinite set of elements called \emph{atomic types} and use the type variable $\tau$ to range over it.

  Intersection types are defined by the following syntax:

  \[\begin{array}{ll@{\qquad}l}

    F, G,\ldots &\recdef \tau \mid A \rightarrow F&\mbox{$F$-types}\\
    A, B,\ldots &\recdef F \mid A \cap B&\mbox{$A$-types}\\
    U, V,\ldots &\recdef A \mid \tempty&\mbox{$U$-types}

  \end{array}\]
  $F$-types are types that are not intersections, $A$-types are types that
  are not empty and $U$-types are types that can be empty.

  We extend the intersection construct as an operation on $U$-types as follows:

  \[
  A \cap \tempty \eqdef A \quad
  \tempty \cap A \eqdef A \quad
  \tempty \cap \tempty \eqdef \tempty
  \]

\end{definition}

\begin{remark}

  For all $U$ and $V$ we have :

  \begin{iteMize}{$\bullet$}

  \item $U \cap V$ exists.

  \item $U \cap \tempty = \tempty \cap U$

  \item If $U \cap V = \tempty$, then $U = V = \tempty$.

  \end{iteMize}

\end{remark}

Note that we do {\bf not} assume any implicit equivalence between intersection
types (such as idempotency, associativity, commutativity).

$F$-types are similar to strict types defined in~\cite{pubsdoc:Bakel-TCS92}.

However, in order to prove theorems such as subject reduction we will need
associativity and commutativity of the intersection $\cap$.
So we define an equivalence relation on types.

\begin{definition}[$\approx$]
  We inductively define $U \approx V$ by the rules of Fig.~\ref{fig:inter-equiv}.
\end{definition}

\begin{figure}[!h]
  \[
  \begin{array}{|c|}

    \upline

    \infer{F \approx F}{\strut} \quad

    \infer{A \cap B \approx B \cap A}{\strut} \quad

    \infer{A \cap B \approx A' \cap B'}{A \approx A' \quad B \approx B'}\quad

    \infer{A \approx C}{A \approx B \quad B \approx C}\quad\\\\

    \infer{(A \cap B) \cap C \approx A \cap (B \cap C)}{}\quad

    \infer{A \cap (B \cap C) \approx (A \cap B) \cap C}{}\quad

    \infer{\tempty \approx \tempty}{}

    \downline

  \end{array}
  \]
  \caption{Equivalence on intersection types}
  \label{fig:inter-equiv}
\end{figure}

The intersection types that we use here differ from those of~\cite{bernadetleng11}, in that the associativity and commutativity (AC) of the intersection $\cap$ are only featured ``on the surface'' of types, and not underneath functional arrows $\rightarrow$.  This will make the typing rules much more syntax-directed, simplifying the proofs of soundness and completeness of typing with respect to the strong normalisation property.  More to the point, this approach reduces the use of the AC properties to the only places where they are needed.

\begin{toappendix}
  \appendixbeyond 0
  \begin{lemma}[Properties of $\approx$]\label{lem:capprop}

    For all $U$, $V$, $W$, $F$, $U'$, $V'$,

    \begin{enumerate}[\em(1)]

    \item $\approx$ is an equivalence relation.

    \item If $U \approx \tempty$, then $U = \tempty$ and if $U \approx F$, then
      $U = F$.

    \item $U \cap V \approx V \cap U$ and
      $(U \cap V) \cap W \approx U \cap (V \cap W)$.

    \item If $U \approx U'$ and $V \approx V'$, then $U \cap V \approx U' \cap V'$.

    \item  For all $U$ and $V$, if $U \cap V \approx U$, then $V = \tempty$.

    \end{enumerate}

  \end{lemma}
\end{toappendix}

\begin{toappendix}
[\begin{proof}See Appendix~\thisappendix.\end{proof}]

  \begin{proof}

    The first four points are proved by straightforward inductions on derivations.
    The last point is more subtle. We define $\phi(U)$ by induction on $U$ as
    follows:

    \[\begin{array}{ll}

      \phi(\tempty) &\eqdef 0\\

      \phi(F) &\eqdef 1\\

      \phi(A \cap B) &\eqdef \phi(A) + \phi(B)

    \end{array}
    \]

    So for all $U$ and $V$ we have $\phi(U \cap V) = \phi(U) + \phi(V)$. 
    Also, for all $A$, $\phi(A) > 0$. So for all $U$, if $\phi(U) = 0$, then
    $U = \tempty$.

    Moreover, for all $U$ and $V$, if $U \approx V$, then
    $\phi(U) = \phi(V)$ (by induction on $U \approx V$).

    Now if $U \cap V \approx U$, then

    \[\phi(U \cap V) = \phi(U) + \phi(V) =  \phi(U)\]

    So we have $\phi(V) = 0$, from which we get $V = \tempty$.
  \end{proof}
\end{toappendix}

We equip intersection types with a notion of sub-typing:

\begin{definition}[$\subseteq$]
  We write $U \subseteq V$ if there exists $U'$ such that $U \approx V \cap U'$.
\end{definition}

\begin{lemma}[Properties of $\subseteq$]

  \label{lem:capsubprop}

  For all $U$, $U'$, $V$, $V'$ :

  \begin{enumerate}[\em(1)]

  \item $\subseteq$ is a partial pre-order for intersection types, and
    $U \approx U'$ if and only if $U\subseteq U'$ and $U'\subseteq U$.

  \item $U \cap V \subseteq U$ and $U \subseteq \tempty$

  \item If $U \subseteq U'$ and $V \subseteq V'$, then
    $U \cap V \subseteq U' \cap V'$

  \end{enumerate}

\end{lemma}

\begin{proof}
  Straightforward with Lemma~\ref{lem:capprop}.
\end{proof}

\subsection{Typing contexts}

We now lift those concepts to typing contexts before presenting the typing
rules.

\begin{definition}[Contexts]\strut

  A context $\Gamma$ is a total map from variables to $U$-types such that
  $\dom\Gamma \eqdef \{ x \mid \Gamma(x) \neq \tempty \} $ is finite. The
  intersection of contexts $\Gamma \cap \Delta$, the relations
  $\Gamma \approx \Delta$ and $\Gamma \subseteq \Delta$, are defined point-wise.

  By $()$ we denote the context mapping every variable to $\tempty$ and by
  $x \col U$ the context mapping $x$ to $U$ and every other variable to $\tempty$.

  The special case of $\Gam \cap \Del$ when $\dom \Gam$ and $\dom \Del$ are
  disjoint is denoted $\Gam, \Del$.
\end{definition}

\begin{lemma}[Properties of contexts]
  For all contexts $\Gamma$, $\Gamma'$, $\Delta$, $\Delta'$, $\Gamma''$,
  \begin{enumerate}[\em\phantom0(1)]
  \item $\Gamma \cap () = \Gamma = () \cap \Gamma$
    (for instance $\Gamma, x\col\tempty = \Gamma = x\col\tempty, \Gamma$)
  \item If $\Gamma \cap \Delta = ()$, then $\Gamma = \Delta = ()$
    and if $\Gamma \approx ()$, then $\Gamma = ()$
  \item $\approx$ is an equivalence relation on contexts.
  \item $\Gamma \cap \Delta \approx \Delta \cap \Gamma$
    and $(\Gamma \cap \Gamma') \cap \Gamma'' \approx
    \Gamma \cap (\Gamma' \cap \Gamma'')$
  \item If $\Gamma \approx \Gamma'$ and $\Delta \approx \Delta'$, then
    $\Gamma \cap \Gamma' \approx \Delta \cap \Delta'$
  \item $\Gamma \subseteq \Delta$ if and only if there exists $\Gamma'$ such
    that $\Gamma \approx \Delta \cap \Gamma'$.
  \item $\subseteq$ is a partial pre-order for contexts, and
    $\Gamma\approx\Delta$ \iff\ $\Gamma\subseteq\Delta$ and
    $\Delta\subseteq\Gamma$.
  \item $\Gamma \cap \Delta \subseteq \Gamma$
  \item If $\Gamma \subseteq \Gamma'$ and $\Delta \subseteq \Delta'$, then
    $\Gamma \cap \Delta \subseteq \Gamma' \cap \Delta'$.
  \item[\em(10)] $(\Gamma, x \col U) \subseteq \Gamma$, in particular
    $\Gamma \subseteq ()$.
  \end{enumerate}
\end{lemma}

\begin{proof}
  The proofs of these properties are straightforward with the use of Lemma~\ref{lem:capprop} and Lemma~\ref{lem:capsubprop}.
\end{proof}

\subsection{Typing judgements}

Now that we have defined types  and contexts we can define typing derivations.
Instead of defining three typing systems for the three calculi we can define
one typing system for the common grammar.

\begin{definition}[Typability in System~\InterLambdaSub]\strut

  The judgement $\Derilis \Gamma M U$ denotes the derivability of $\Derili \Gamma M U$ with the rules of Fig.~\ref{fig:InterLambdaSub}.  We write $\Derilis[n] \Gamma M U$ if there exists a derivation with $n$ uses of the (App) rule.
\end{definition}

\begin{figure}[h]
  \[\begin{array}{|c|}

    \upline

    \infer{\Derili {x \col F} x F}{\strut}\text{(Var)} \qquad

    \infer{\Derili \Gamma {\l x . M} {A \rightarrow F}}
    {\Derili {\Gamma, x \col U} M F \quad A \subseteq U}\text{(Abs)}\qquad

    \infer{\Derili {\Gamma \cap \Delta} {M N} F}
    {\Derili \Gamma M {A \rightarrow F} \quad \Derili \Delta N A}\text{(App)}\\\\

    \infer{\Derili {\Gamma \cap \Delta} M {A \cap B}}
    {\Derili \Gamma M A \quad \Derili \Delta M B}\text{(Inter)} \qquad

    \infer{\Derili {} M \tempty}{\strut}\text{(Omega)}\\\\

    \infer{\Derili {\Gamma \cap \Delta} {\Subst M x N} F}
    {\Derili \Gamma N A \quad \Derili {\Delta, x : U} M F \quad
      U = A \vee U = \tempty}\text{(Subst)}\\\\


    \infer{\Derili {\Gamma, x \col U \cap (V_1 \cap V_2)} {\contraction x y z M} F}
    {\Derili {\Gamma, x \col U, y \col V_1, z \col V_2} M F}\text{(Contraction)}
    \qquad

    \infer{\Derili {\Gamma, x \col U \cap A} {\weakening x M} F}
    {\Derili {\Gamma, x \col U} M F}\text{(Weakening)}

    \downline
  \end{array}
  \]
  \caption{System \InterLambdaSub}
  \label{fig:InterLambdaSub}
\end{figure}

Note that the rule deriving $\Derili {} M \tempty$ does not interfere with the
rest of the system as $\tempty$ is not an $A$-type. 
It is only here for convenience to synthetically express some statements and
proofs that would otherwise need a verbose case analysis
(\eg Lemma~\ref{lem:typsubSR}).

Examples of how $\l$-terms are typed are given in the next section.

Also, note that the introduction rule for the intersection is directed by the
syntax of the types: 
If $\Derilis {\Gamma\cap\Delta} M {A \cap B}$, then the last rule of the derivation is necessarily $(Inter)$ and its premises are necessarily $\Derili {\Gamma} M A$ and $\Derili{\Delta} M B$.
We are not aware of any intersection type system featuring this property, which
is here a consequence of dropping the implicit AC properties of intersections,
and a clear advantage over the system in~\cite{bernadetleng11}.
Similar properties can however be found in systems such as those of~\cite{GhilezanILL11}, which avoids having to explicitly type a term by an intersection type.

\begin{lemma}[Basic properties of \InterLambdaSub]\strut\label{lem:basicprop}

  \begin{enumerate}[\em(1)]

  \item If $\Derilis[n] \Gamma M {U \cap V}$, then there exist $\Gamma_1$,
    $\Gamma_2$, $n_1$, $n_2$ such that $n = n_1 + n_2$,
    $\Gamma = \Gamma_1 \cap \Gamma_2$, $\Derilis[n_1] {\Gamma_1} M U$ and
    $\Derilis[n_2]{\Gamma_2} M V$.

  \item If $\Derilis \Gamma M {A}$, then $\dom \Gam=\FV M$. 

  \item If $\Derilis[n] \Gamma M U$ and $U \approx U'$, then there exists $\Gamma'$
    such that $\Gamma \approx \Gamma'$ and $\Derilis[n] {\Gamma'} M {U'}$

  \item If $\Derilis[n] \Gamma M U$ and $U \subseteq V$, then there exist $m$ and
    $\Delta$ such that $m \leq n$, $\Gamma \subseteq \Delta$ and
    $\Derilis[m] \Delta M V$

  \end{enumerate}
\end{lemma}

\begin{proof}
  The first point generalises to $U$-types the previous remark. The second point is by induction on the typing tree. The third point is by induction on the derivation of $U \approx U'$, and the fourth one combines the previous points.
\end{proof}

The following lemma is used to prove Subject Reduction in both $\LambdaS$ and
$\LambdaLxr$:

\begin{toappendix}
  \begin{lemma}[Typing of explicit substitution] \label{lemma:TypingSubst}
    Assume $\Derilis[n] {\Gamma, x \col A} M B$ and $\Derilis[m] \Delta N A$.
    Then, there exists $\Gamma'$ such that $\Gamma' \approx \Gamma \cap \Delta$ and $\Derilis[n + m] {\Gamma'} {\Subst M x N} B$.
  \end{lemma}
\end{toappendix}

\begin{toappendix}
  [\begin{proof}See Appendix~\thisappendix.\end{proof}]
  \begin{proof}

    By induction on $B$:

    \begin{iteMize}{$\bullet$}

    \item If $B = F$, then the result is trivial : we use the (Subst) rule.

    \item If $B = B_1 \cap B_2$, then, by Lemma~\ref{lem:basicprop}.1, there
			exist $\Gamma_1$, $\Gamma_2$, $A_1$,
      $A_2$, $n_1$ and $n_2$ such that $\Gamma = \Gamma_1 \cap \Gamma_2$,
      $A = A_1 \cap A_2$, $n = n_1 + n_2$,
      $\Derili[n_1] {\Gamma_1, x \col A_1} M {B_1}$ and
      $\Derili[n_2] {\Gamma_2, x \col A_2} M {B_2}$.
			By hypothesis, $\Derilis[m] \Delta N A$.
      Hence, by Lemma~\ref{lem:basicprop}.1, there exist $\Delta_1$, $\Delta_2$,
			$m_1$ and $m_2$ such that
      $\Delta = \Delta_1 \cap \Delta_2$, $m = m_1 + m_2$,
      $\Derili[m_1] {\Delta_1} N {A_1}$ and $\Derili[m_2] {\Delta_2} N {A_2}$.

      By induction hypothesis, there exist $\Gamma_1'$ and $\Gamma_2'$ such that
      $\Gamma_1' \approx \Gamma_1 \cap \Delta_1$,
      $\Gamma_2' \approx \Gamma_2 \cap \Delta_2$,
      $\Derili[n_1 + m_1] {\Gamma_1'} {\Subst M x N} {B_1}$ and
      $\Derili[n_2 + m_2] {\Gamma_2'} {\Subst M x N} {B_2}$.
      So we have $\Derili[n_1 + m_1 + n_2 + m_2] {\Gamma_1' \cap \Gamma_2'}
      {\Subst M x N} {B_1 \cap B_2}$ with $n_1 + m_1 + n_2 + m_2 = n + m$,
      $\Gamma_1' \cap \Gamma_2' \approx
      (\Gamma_1 \cap \Delta_1) \cap (\Gamma_2 \cap \Delta_2) \approx
      (\Gamma_1 \cap \Gamma_2) \cap (\Delta_1 \cap \Delta_2) \approx
      \Gamma \cap \Delta$ and $B = B_1 \cap B_2$.

    \end{iteMize}

  \end{proof}
\end{toappendix}

\begin{remark}

  The previous theorem is not true if we replace $A$ by $\tempty$:
  If $B$ is an intersection, we would need to duplicate the typing tree of $N$.

\end{remark}

\section{Soundness}
\label{sec:soundness}

In this section we prove Subject Reduction and Soundness: respectively, the property that typing is preserved by reduction, and the property that if a term is typable, then it is strongly normalising for the reduction relation.  This is true for the three calculi, but specific to each of them because the reduction relation is itself specific to each calculus.

Therefore in this section, we work separately on each calculus: each time, we define the reduction rules and prove the Subject Reduction property, which leads to Soundness.

\subsection{Pure $\l$-calculus}

\label{sec:SNLambda}

Remember that a pure $\l$-term is a term $M$ that does not contain any
explicit substitutions $\Subst M x N$, or weakenings $\weakening x M$ or
contractions $\contraction x y z M$.

As we will see, only strongly normalising terms can be assigned an $A$-type by
the system (Theorem~\ref{th:soundness}).
In fact, all of them can (Theorem~\ref{th:lambda-compl}), see for instance how
the example below correctly uses the abstraction rule ($A\subseteq \tempty$).

\[\infer{\Derili{x\col F} {\l y.x}{A\arr F}}
{
  \infer{\Derili{x\col F,y\col \tempty} {x}{F}}
  {}
}
\]

Owing to non-idempotency, no closed term inhabits the simple type
$(\tau\arr \tau\arr \tau')\arr(\tau\arr \tau')$ (with $\tau\neq\tau'$), but its
natural inhabitant $\l f.\l x.f\ x\ x$  in a simply-typed system can here be
given type $(\tau\arr \tau\arr \tau')\arr(\tau\cap \tau\arr \tau')$.

\begin{definition}[Reduction in $\l$-calculus]\strut\\
  If $M$ and $N$ are pure $\l$-terms, we denote by $\subst M x N$ the result of the (implicit) substitution (as defined in \eg\cite{Bar84}).

  The reduction rule is $\beta$-reduction:
  \[(\l x.M)\ N \Rew{} \subst M x N\]
  The congruent closure of this rule is denoted $\Rew{\beta}$.
  $\SNLambda$ denotes the set of strongly normalising $\l$-terms
  (for $\beta$-reduction).

\end{definition}

\begin{toappendix}
  \begin{lemma}[Typing of implicit substitutions]
    \label{lem:typsubSR}
    If $\Derilis[n] {\Gamma, x \col U} M A$ and $\Derilis[m] \Delta N U$, then there exists $\Gamma'$ such that $\Gamma' \approx \Gamma \cap \Delta$ and $\Derilis[n + m] {\Gamma'} {\subst M x N} A$.
  \end{lemma}
\end{toappendix}

\begin{toappendix}
  [\begin{proof}See Appendix~\thisappendix.\end{proof}]
  \begin{proof}

    By induction on the derivation of $\Deri {\Gamma, x : U} M A$.

    \begin{iteMize}{$\bullet$}

    \item $\infer{\Deri {x : F} x F}{}$\\
      Here $\Gamma = ()$, $x = M$, $n = 0$, $A = F$ and $U = F$.
      Therefore, $\subst M x N = N$ and $\Gamma \cap \Delta = \Delta$.
			By hypothesis, $\Deri[m] \Delta N U$.
      So we have $\Deri[n + m] {\Gamma \cap \Delta} {\subst M x N} A$ with
      $\Gamma \cap \Delta \approx \Gamma \cap \Delta$.\bigskip

    \item $\infer{\Deri {y : F} y F}{}$ with $y \neq x$\\
      Here $\Gamma = (y : F)$, $A = F$, $U = \tempty$, $M = y$, $n = 0$.
      Since $\tempty = U$ and $\Deri[m] \Delta N U$, we have $\Delta = ()$ and
			$m = 0$.
      Then, we have, $\Deri[n + m] {\Gamma \cap \Delta} {\subst M x N} A$
			because $\subst M x N = y$.\bigskip

    \item $\infer{\Deri[n_1 + n_2] {\Gamma_1 \cap \Gamma_2, x : U_1 \cap U_2}
        M {A_1 \cap A_2}}{\Deri[n_1] {\Gamma_1, x : U_1} M {A_1} \quad \Deri[n_2]
        {\Gamma_2, x : U_2} M {A_2}}$\\
      Here $A = A_1 \cap A_2$, $n = n_1 + n_2$, $U = U_1 \cap U_2$ and
      $\Gamma = \Gamma_1 \cap \Gamma_2$.
			By hypothesis, $\Deri[m] \Delta N U$.
      So, by Lemma~\ref{lem:basicprop}.1, there exist $\Delta_1$, $\Delta_2$,
			$m_1$, $m_2$ such that
      $\Delta = \Delta_1 \cap \Delta_2$, $\Deri[m_1] {\Delta_1} N {U_1}$ and
      $\Deri[m_2] {\Delta_2} N {U_2}$.
      By induction hypothesis, there exists $\Gamma'_1$ such that
      $\Gamma'_1 \approx \Gamma_1 \cap \Delta_1$ and
      $\Deri[n_1 + m_1] {\Gamma'_1} {\subst M x N} {A_1}$.
      We also have the existence of $\Gamma'_2$ such that
      $\Gamma'_2 \approx \Gamma_2 \cap \Delta_2$ and
      $\Deri[n_2 + m_2]{\Gamma'_2} {\subst M x N} {A_2}$. 

      Therefore we have $\Deri[n_1 + m_1 + n_2 + m_2]{\Gamma'_1 \cap \Gamma_2'} M A$
      with $n_1 + m_1 + n_2 + m_2 = n + m$ and $\Gamma'_1 \cap \Gamma'_2 \approx
      (\Gamma_1 \cap \Delta_1) \cap (\Gamma_2 \cap \Delta_2) \approx
      (\Gamma_1 \cap \Gamma_2) \cap (\Delta_1 \cap \Delta_2) \approx \Gamma
      \cap \Delta$.\bigskip

    \item $\infer{\Deri[n_1 + n_2 + 1] {\Gamma_1 \cap \Gamma_2, x : U_1 \cap U_2}
        M {A}}{\Deri[n_1] {\Gamma_1, x : U_1} {M_1} {A_1 \rightarrow A} \quad \Deri[n_2]
        {\Gamma_2, x : U_2} {M_2} {A_1}}$\\
      Here $M = M_1 M_2$, $n = n_1 + n_2 + 1$, $U = U_1 \cap U_2$ and
      $\Gamma = \Gamma_1 \cap \Gamma_2$.
			By hypothesis, $\Deri[m] \Delta N U$.
      So, by Lemma~\ref{lem:basicprop}, there exist $\Delta_1$, $\Delta_2$,
			$m_1$, $m_2$ such that
      $\Delta = \Delta_1 \cap \Delta_2$, $m = m_1 + m_2$,
			$\Deri[m_1] {\Delta_1} N {U_1}$ and
      $\Deri[m_2] {\Delta_2} N {U_2}$.
      By induction hypothesis, there exists $\Gamma'_1$ such that
      $\Gamma'_1 \approx \Gamma_1 \cap \Delta_1$ and
      $\Deri[n_1 + m_1] {\Gamma'_1} {\subst {M_1} x N}{A_1 \rightarrow A}$.
      We also have the existence of $\Gamma'_2$ such that
      $\Gamma'_2 \approx \Gamma_2 \cap \Delta_2$ and
			$\Deri[n_2 + m_2] {\Gamma'_2}
      {\subst {M_2} x N}{A_1}$. 

      Therefore we have $\Deri[n_1 + m_1 + n_2 + m_2 + 1]
      {\Gamma'_1 \cap \Gamma_2'} {\subst {M_1 M_2} x N} A$ with
			$n_1 + m_1 + n_2 + m_2 + 1 = n + m$ and
      $\Gamma'_1 \cap \Gamma'_2 \approx (\Gamma_1 \cap \Delta_1) \cap
      (\Gamma_2 \cap \Delta_2) \approx (\Gamma_1 \cap \Gamma_2) \cap
      (\Delta_1 \cap \Delta_2) \approx \Gamma \cap \Delta$.\bigskip

    \item $\infer{\Deri[n] {\Gamma, x \col U} {\l y . M_1} {A_1 \rightarrow F}}
      {\Deri[n] {\Gamma, x \col U, y \col V} {M_1} F \quad {A_1} \subseteq V}$
      with $x \neq y$ and $y \notin fv(N)$.\\

      \noindent By induction hypothesis, there exist $\Gamma'$ such that
      $(\Gamma, y \col V) \cap \Delta \approx \Gamma'$ and
      $\Deri[n + m] {\Gamma'} {\subst {M_1} x N} F$.
      $y \notin fv(N)$, so there exist $\Gamma''$ and $V'$ such that
      $\Gamma' \approx (\Gamma'', y \col V')$, $V \approx V'$, and
      $\Gamma'' \approx \Gamma \cap \Delta$. Then we can conclude.\bigskip

    \item We do not have to deal with the other rules because they cannot be used
      for a pure $\l$-term.

    \end{iteMize}

  \end{proof}
\end{toappendix}

\begin{toappendix}
  \begin{theorem}[Subject Reduction for $\l$]
    \label{th:SubjectReduction}
    If $\Derilis[n] {\Gamma} M A$ and $M \Rew{\beta} M'$, then there exist $m$ and
    $\Delta$ such that $m < n$, $\Gamma \subseteq \Delta$ and
    $\Derilis[m] \Delta {M'} A$.
  \end{theorem}
\end{toappendix}

\begin{toappendix}
  [\begin{proof}See Appendix~\thisappendix.\end{proof}]
  \begin{proof}

    First by induction on $M \Rew{\beta} M'$, then by induction on $A$.

    \begin{iteMize}{$\bullet$}

    \item If there exist $A_1$ and $A_2$ such that $A = A_1 \cap A_2$, then,
			by Lemma~\ref{lem:basicprop}.1, there
      exist $\Gamma_1$, $\Gamma_2$, $n_1$, $n_2$ such that
      $\Gamma = \Gamma_1 \cap \Gamma_2$, $n = n_1 + n_2$, $\Deri[n_1] {\Gamma_1}
      M {A_1}$ and $\Deri[n_2] {\Gamma_2} M {A_2}$.
      By induction hypothesis (on $(M \Rew{\beta} M', A_1)$ and
      $(M \Rew{\beta} M', A_2)$), there exist $\Delta_1$ and $m_1$ such that
      $\Gamma_1 \subseteq \Delta_1$, $m_1 < n_1$ and
      $\Deri[m_1] {\Delta_1} {M'} {A_1}$.
      We also have the existence of $\Delta_2$ and $m_2$ such that $m_2 < n_2$,
      $\Gamma_2 \subseteq \Delta_2$ and $\Deri[m_2] {\Delta_2} {M'} {A_2}$.
      Then we have $\Deri[m_1 + m_2] {\Delta_1 \cap \Delta_2} {M'} A$ with
      $m_1 + m_2 < n$ and $\Gamma \subseteq \Delta_1 \cap \Delta_2$.

    \item $\infer{(\l x . M_1) M_2 \Rew{\beta} \subst {M_1} x {M_2}}{}$ and
      $A = F$\\
      Then there exist $\Gamma_1$, $\Gamma_2$, $n_1$, $n_2$ and $B$ such that
      $\Gamma = \Gamma_1 \cap \Gamma_2$, $n = n_1 + n_2 + 1$, $\Deri[n_1] {\Gamma_1}
      {\l x . M_1} {B \rightarrow F}$ and $\Deri[n_2] {\Gamma_2} {M_2} B$.
      So there exists $U$ such that $B \subseteq U$ and $\Deri[n_1] {\Gamma_1, x : U}
      {M_1} F$.
      Then, by Lemma~\ref{lem:basicprop}.4, there exist $\Gamma_2'$ and $n_2'$
			such that
      $\Gamma_2 \subseteq \Gamma_2'$, $n_2' \leq n_2$ and
      $\Deri[n_2']{\Gamma_2'} {M_2} U$.
      Therefore, by the substitution lemma (Lemma~\ref{lem:typsubSR}), there exists $\Gamma'$ such that
      $\Gamma' \approx \Gamma_1 \cap \Gamma_2'$ and
      $\Deri[n_1 + n_2'] {\Gamma'} {\subst {M_1} x {M_2}} F$ with $A = F$,
      $n_1 + n_2' < n$ and $\Gamma = \Gamma_1 \cap \Gamma_2 \subseteq
      \Gamma_1 \cap \Gamma_2' \approx \Gamma'$.

    \item $\infer{\l x . M_1 \Rew{\beta} \l x . M_1'}{M_1 \Rew{\beta} M_1'}$ and
      $A = F$\\
      Here, there exist $B$, $G$ and $U$ such that $A = B \rightarrow G$,
      $B \subseteq U$ and $\Deri[n] {\Gamma, x : U} {M_1} G$.
      By induction hypothesis, there exist $\Gamma'$, $U'$ and $m$ such that
      $\Gamma \subseteq \Gamma'$, $U \subseteq U'$, $m < n$ and
      $\Deri[m] {\Gamma', x : U'} {M_1'} G$.
      So we have $B \subseteq U'$. Hence $\Deri[m] {\Gamma'} {M'} A$.

    \item $\infer{M_1 M_2 \Rew{\beta} M_1' M_2}{M_1 \Rew{\beta} M_1'}$ and $A = F$\\
      Then there exist $B$, $\Gamma_1$, $\Gamma_2$, $n_1$, $n_2$ such that
      $n = n_1 + n_2 + 1$, $\Gamma = \Gamma_1 \cap \Gamma_2$, $\Deri[n_1]
      {\Gamma_1} {M_1} {B \rightarrow A}$ and $\Deri[n_2] {\Gamma_2} {M_2} B$.
      By induction hypothesis there exist $\Gamma_1'$ and $n_1'$ such that
      $\Gamma_1 \subseteq \Gamma_1'$, $n_1' < n_1$ and $\Deri[n_1'] {\Gamma_1'} {M_1'}
      {B \rightarrow F}$.
      Therefore $\Deri[n_1' + n_2 + 1] {\Gamma_1' \cap \Gamma_2} {M_1' M_2} A$ with
      $n_1' + n_2 + 1 < n$ and $\Gamma \subseteq \Gamma_1' \cap \Gamma_2$.

    \item $\infer{M_1 M_2 \Rew{\beta} M_1 M_2'}{M_2 \Rew{\beta} M_2'}$ and $A = F$\\
      Then there exist $B$, $\Gamma_1$, $\Gamma_2$, $n_1$, $n_2$ such that
      $n = n_1 + n_2 + 1$, $\Gamma = \Gamma_1 \cap \Gamma_2$, $\Deri[n_1]
      {\Gamma_1} {M_1} {B \rightarrow A}$ and $\Deri[n_2] {\Gamma_2} {M_2} B$.
      By induction hypothesis there exist $\Gamma_2'$ and $n_2'$ such that
      $\Gamma_2 \subseteq \Gamma_2'$, $n_2' < n_2$ and
      $\Deri[n_2'] {\Gamma_2'} {M_2'} B$.
      Therefore, $\Deri[n_1 + n_2' + 1] {\Gamma_1 \cap \Gamma_2'} {M_1 M_2'} A$ with
      $n_1 + n_2' + 1 < n$ and $\Gamma \subseteq \Gamma_1 \cap \Gamma_2'$.\qedhere
    \end{iteMize}
  \end{proof}
\end{toappendix}

The above theorem and its proof are standard but for the quantitative
information in the typability properties.
This is where non-idempotent intersections provide a real advantage over
idempotent ones, as every $\beta$-reduction strictly reduces the number of
application rules in the typing trees: no sub-tree is duplicated in the process.
This is something specific to non-idempotent intersection types, and obviously
false for simple types or idempotent intersection types.

As a direct corollary we obtain:

\begin{theorem}[Soundness for $\l$]  \label{th:soundness}
  If $M$ is a pure $\l$-term and $\Derilis \Gamma M A$, then $M\in \SNLambda$.
\end{theorem}

The converse is also true (strongly normalising terms can be typed in
\InterLambdaSub), see Theorem~\ref{th:lambda-compl} and more generally section
\ref{lambda-compl} (with Subject Expansion, \etc).

\subsection{$\LambdaS$}

Remember that terms of $\LambdaS$ are terms that do not contain any weakenings
$\weakening x M$ or contractions $\contraction x y z M$.
In other words, we consider the extension of the pure $\l$-calculus with
explicit substitutions ($\Subst M x N$). This is the same syntax as that of
\lx~\cite{BlooRose:csn1995}, but unlike \lx, the reduction rules only duplicate
substitutions when needed.
For example, the following rule:
\[ \Subst {(M_1 M_2)} x N \Rew{} \Subst {M_1} x N \Subst {M_2} x N \]
can only be applied  if $x \in fv(M_1)$ and $x \in fv(M_2)$.

In the other cases, the explicit substitution will only go one way.
The rules are chosen with the proof of Subject Reduction in mind, which is a
simple adaptation of the proof in the pure $\l$-calculus.
This leads to soundness (typable implies strongly normalising), and therefore
Melliès's counter-example~\cite{Mellies:tlca1995} to strong normalisation is
naturally avoided.

\begin{definition}[Reduction in $\LambdaS$]\strut

  The reduction and equivalence rules of $\LambdaS$ are presented in Fig.~\ref{fig:LambdaSRew}.

  For a set of rules $E \subseteq \{B, S, W\}$ from
  Figure~\ref{fig:LambdaSRew}, $\Rew{E}$ denotes the congruent closure of
  the rules in $E$ modulo the $\equiv$ rule.

  $\SNLambdaS$ denotes the set of strongly normalising $\LambdaS$-terms for
  $\Rew{B,S,W}$.
\end{definition}

\begin{figure}
  \begin{small}
    \[\begin{array}{|llcll|}

      \upline

      B: & (\l x . M) N & \Rew{} & \Subst  M x N & \\
      W: & \Subst y x N & \Rew{} & y & x \neq y\\

      S: & & & &  \\

      & \Subst x x N & \Rew{} & N & (SR) \\

      & \Subst {(M_1 M_2)} x N & \Rew{} & (\Subst {M_1} x N) (\Subst {M_2} x N) &
      x \in fv(M_1), x \in fv(M_2)\\

      & \Subst {(M_1 M_2)} x N & \Rew{} & M_1 (\Subst {M_2} x N) &
      {x \notin fv(M_1), x \in fv(M_2)}\\

      & \Subst {(M_1 M_2)} x  N & \Rew{} & (\Subst {M_1} x N) M_2 &
      x \notin fv(M_2)\\

      & \Subst {(\l y . M)} x N & \Rew{} & \l y . \Subst M x N & 
      x \neq y, y \notin fv(N)\\

      & \Subst {(M_1[y := M_2])} x N & \Rew{} & \Subst {(\Subst {M_1} x N)} y
      {\Subst {M_2} x N} &
      x \in fv(M_1), x \in fv(M_2), y \notin fv(N)\\

      & \Subst {(\Subst {M_1} y {M_2})} x  N & \Rew{} & \Subst {M_1} y
      {\Subst {M_2} x N} &
      x \notin fv(M_1), x \in fv(M_2)\\

      & \Subst {(\Subst M x {N_1})} y  {N_2} & \equiv &
      \Subst {(\Subst {M} y {N_2})} x {N_1} &
      x \neq y, x \notin fv(N_2), y \notin fv(N_1)

      \downline

    \end{array}\]

  \end{small}

  \caption{Reduction and equivalence rules of $\LambdaS$}
  \label{fig:LambdaSRew}

\end{figure}

We call this calculus $\LambdaS$ because it is a variant of the calculi $\lambda_s$ of~\cite{Kesner07} and $\lambda_{es}$ of~\cite{RenaudPhD}.
That of~\cite{RenaudPhD} is more general than that of~\cite{Kesner07} in the sense that it allows the reductions
\[\begin{array}{llcll}
  \!\!\!(1) & \Subst {(M_1 M_2)} x N &\Rew{}& \Subst {M_1} x N M_2 & \mbox{ when }x \notin fv(M_1), x \notin fv(M_2)\\
  \!\!\!(2) & \Subst {(M_1 M_2)} x N &\Rew{}& M_1 \Subst {M_2} x N & \mbox{ when }x \notin fv(M_1), x \notin fv(M_2)\hbox to 68 pt{\hfill}
\end{array}
\]

Reduction $(2)$ is problematic in our approach since, even though the Subject Reduction property would still hold, it would not hold with the quantitative information from which Strong Normalisation can be proved: In the typing tree, the type of $M_1$ is not an intersection (it is an $F$-type) but the type of $M_2$ can be one.  So we cannot directly type $\Subst {M_2} x N$.  If $x \in fv(M_2)$ we can use Lemma~\ref{lemma:TypingSubst}, otherwise we have to duplicate the typing tree of $N$.

We therefore exclude $(2)$ from the calculus, but keep $(1)$ as one of our rules, since it is perfectly compatible with our approach. It is also needed to simulate (in several steps) the general \emph{garbage collection} rule below
\[ \Subst M x N \Rew{} M \quad (x \notin fv(M)) \]
which is present in both~\cite{Kesner07} and~\cite{RenaudPhD}, and which we decide to restrict, for simplicity, to the case where $M$ is a variable different from $x$.\footnote{\cite{Kesner07} needs the general version, if only for the lack of rule $(1)$.} All of our results would still hold with the general garbage collection rule.

\begin{toappendix}
  \appendixbeyond 0
  \begin{lemma}
\label{lem:SNForSW}
    $\Rew{S, W}$ terminates.
  \end{lemma}
\end{toappendix}

\begin{toappendix}
  [\begin{proof}By a polynomial argument. More precisely, see appendix \thisappendix.\end{proof}]

  \begin{proof}

    By a polynomial argument.
    We define $m_x(M)$ as follow: if $x \notin fv(M)$, then
    $m_x(M) = 1$. Otherwise we have:
    \[\begin{array}{lllr}

      m_x(x) & = & 1 & \\

      m_x(\l y . M) & = & m_x(M) & \\

      m_x(M_1 M_2) & = & m_x(M_1) + m_x(M_2) & x \in fv(M_1), x \in fv(M_2) \\

      m_x(M_1 M_2) & = & m_x(M_1) & x \notin fv(M_2) \\

      m_x(M_1 M_2) & = & m_x(M_2) & x \notin fv(M_1) \\

      m_x(\Subst M y N) & = & m_x(M) + m_y(M) \times (m_x(N) + 1) &
      x \in fv(M), x \in fv(N) \\

      m_x(\Subst M y N) & = & m_y(M) \times (m_x(N) + 1) &
      x \notin fv(M), x \in fv(N) \\

      m_x(\Subst M y N) & = & m_x(M) & x \notin fv(N)

    \end{array}
    \]
    We also define $S(M)$ as follow:

    \[\begin{array}{lll}

      S(x) & = & 1 \\

      S(M_1 M_2) & = & S(M_1) + S(M_2) \\

      S(\l x . M) & = & S(M) \\

      S(\Subst M x N) & = & S(M) + m_x(M) \times S(N) \\

    \end{array}
    \]
    Finally, we define $I(M)$ as follow:

    \[\begin{array}{lll}

      I(x) & = & 2 \\

      I(\l x . M) & = & 2 I(M) + 2 \\

      I(M_1 M_2) & = & 2 I(M_1) + 2 I(M_2) + 2 \\

      I(\Subst M x N) & = & I(M) \times (I(N) + 1)

    \end{array}
    \]
    If we consider $n = (S(M), I(M))$ in lexical order, then $\Rew{S,W}$ strictly decreases $n$ and $\equiv$ does not change it.

    Hence $\Rew{S,W}$ terminates.
    This lemma and proof are a special case of~\cite{KesnerR11}.
  \end{proof}
\end{toappendix}

\begin{toappendix}
  \begin{theorem}[Subject Reduction for $\LambdaS$]\strut
    \label{th:SubjectReduction-lS}
    
    Assume $\Derilis[n] \Gamma M A$. We have the following properties:
    
    \begin{iteMize}{$\bullet$}
      
    \item If $M \Rew{B} M'$, then there exist $\Gamma'$ and $m$ such that
      $\Gamma \subseteq \Gamma'$, $m < n$ and $\Derilis[m] {\Gamma'} {M'} A$

    \item If $M \Rew{S} M'$, then there exists $\Gamma'$ such that
      $\Gamma \approx \Gamma'$ and $\Derilis[n] {\Gamma'} {M'} A$

    \item If $M \Rew{W} M'$, then there exist $\Gamma'$ and $m$ such that
      $\Gamma \subseteq \Gamma'$, $m \leq n$ and $\Derilis[m] {\Gamma'} {M'} A$

    \item If $M \equiv M'$, then there exists $\Gamma'$ such that
      $\Gamma \approx \Gamma'$ and $\Derilis[n] {\Gamma'} {M'} A$

    \end{iteMize}

  \end{theorem}
\end{toappendix}

\begin{toappendix}
  [\begin{proof}See Appendix~\thisappendix.\end{proof}]
  
  \begin{proof}

    First by induction on $M \Rew{E} M'$ and $M \equiv M'$, then by
    induction on $A$.

    For modularity,
    the triplet $(\Rew{E}, R, r)$ can be one of the following triplets:
    \[(\Rew{B}, \approx, <),\enspace(\Rew{S}, \approx, =),\enspace
    (\Rew{W}, \subseteq, \leq)), (\equiv, \approx, =).\]


    \begin{iteMize}{$\bullet$}

    \item If $A = A_1 \cap A_2$, then there exist $\Gamma_1$, $\Gamma_2$, $n_1$ and
      $n_2$ such that $\Gamma = \Gamma_1 \cap \Gamma_2$, $n = n_1 + n_2$,
      $\Derilis[n_1] {\Gamma_1} M {A_1}$ and $\Derili[n_2] {\Gamma_2} M {A_2}$.

      By induction hypothesis (on $(M \Rew E M', A_{1})$ and $(M \Rew E M', A_2)$),
      there exist $\Gamma_1'$, $\Gamma_2'$, $m_1$ and $m_2$
      such that $\Gamma_1 ~ R ~ \Gamma_1'$, $\Gamma_2 ~ R ~ \Gamma_2'$,
      $m_1 ~ r ~ n_1$, $m_2 ~ r ~ n_2$, $\Derilis[m_1] {\Gamma_1'} {M'} {A_1}$ and
      $\Derilis[m_2] {\Gamma_2'} {M'} {A_2}$.

      Hence, $\Derilis[m_1 + m_2] {\Gamma_1' \cap \Gamma_2'} {M'} {A_1 \cap A_2}$
      with $A = A_1 \cap A_2$, $m_1 + m_2 ~ r ~ n$,
      $\Gamma ~ R ~ \Gamma_1' \cap \Gamma_2'$.

    \item $(\l x . M_1) M_2 \Rew{B} \Subst {M_1} x {M_2}$ and $A = F$ :
      There exist $\Gamma_1$, $\Gamma_2$, $n_1$, $n_2$ and $B$ such that
      $\Gamma = \Gamma_1 \cap \Gamma_2$, $n = n_1 + n_2 + 1$,
      $\Derilis[n_1] {\Gamma_1} {\l x . M_1} {B \rightarrow F}$ and
      $\Derilis[n_2] {\Gamma_2} {M_2} B$. Hence, there exists $U$ such that
      $B \subseteq U$ and $\Derilis[n_1] {\Gamma_1, x \col U} {M_1} F$.

      \begin{iteMize}{$-$}

      \item If $U = C$, then by using Lemma~\ref{lem:basicprop}.4 there exist $\Delta$ and $m$ such that $m \leq n_2$,
        $\Gamma_2 \subseteq \Delta$ and $\Derilis[m] \Delta {M_2} C$.
        Hence $\Derilis[n_1 + m] {\Gamma_1 \cap \Delta} {\Subst {M_1} x {M_2}} F$ with
        $n_1 + m < n$, $F = A$ and $\Gamma \subseteq \Gamma_1 \cap \Delta$.

      \item If $U = \tempty$, then $\Derilis[n_1 + n_2] {\Gamma_1 \cap \Gamma_2}
        {\Subst {M_1} x {M_2}} F$ with $n_1 + n_2 < n$, $A = F$ and
        $\Gamma \subseteq \Gamma_1 \cap \Gamma_2$.

      \end{iteMize}

    \item $\Subst y x N \Rew{W} y$, $x \neq y$ and $A = F$ :
      there exist $\Gamma_1$, $\Gamma_2$, $n_1$, $n_2$ and $B$ such that
      $\Gamma = \Gamma_1 \cap \Gamma_2$, $n = n_1 + n_2$,
      $\Derilis[n_1] {\Gamma_1} N B$ and
      $\Derilis[n_2] {\Gamma_2, x \col \tempty} y F$.
      So we have $\Derilis[n_2] {\Gamma_2} y A$ with
      $\Gamma \subseteq (\Gamma_2, x \col \tempty)$ and $n_2 \leq n$.

    \item $\Subst x x N \Rew{S} N$ and $A = F$ :
      there exist $\Gamma_1$, $\Gamma_2$, $n_1$, $n_2$ and $B$ such that
      $\Gamma = \Gamma_1 \cap \Gamma_2$, $n = n_1 + n_2$,
      $\Derilis[n_1] {\Gamma_1} N B$ and
      $\Derilis[n_2] {\Gamma_2, x \col B} x F$.
      Hence $\Gamma_2 = ()$ and $B = F$ and $n_2 = 0$.
      Therefore $\Gamma = \Gamma_1$ and $n_1 = n$.
      So we have $\Derilis[n] {\Gamma} N A$ with $\Gamma \approx \Gamma$.

    \item $\Subst {(M_1 M_2)} x N \Rew{S} \Subst {M_1} x N \Subst {M_2} x N$ with
      $x \in fv(M_1)$, $x \in fv(M_2)$ and $A = F$:
      Then there exist $\Gamma_1$, $\Gamma_2$, $\Gamma_3$, $\Gamma_4$,
      $A_1$, $A_2$, $B$, $n_1$, $n_2$, $n_3$ and $n_4$ such that:
      $\Derilis[n_1] {\Gamma_1, x \col A_1} {M_1} {B \rightarrow F}$
      $\Derilis[n_2] {\Gamma_2, x \col A_2} {M_2} B$,
      $\Derilis[n_3] {\Gamma_3} N {A_1}$, $\Derilis[n_4] {\Gamma_4} N {A_2}$,
      $n = n_1 + n_2 + n_3 + n_4$ and
      $\Gamma = (\Gamma_1 \cap \Gamma_2) \cap (\Gamma_3 \cap \Gamma_4)$.
      Hence $\Derilis[n_1 + n_3 + n_2 + n_4]
      {(\Gamma_1 \cap \Gamma_3) \cap (\Gamma_2 \cap \Gamma_4)}
      {\Subst {M_1} x N \Subst {M_2} x N} F$.

    \item $\Subst {(M_1 M_2)} x N  \Rew{S} \Subst {M_1} x N M_2$ with
      $x \notin fv(M_2)$ and $A = F$:
      Then there exist $\Gamma_1$, $\Gamma_2$, $\Gamma_3$, $U$, $A_1$, $B$, $n_1$,
      $n_2$, and $n_3$ such that:
      $\Derilis[n_1] {\Gamma_1, x \col U} {M_1} {B \rightarrow F}$
      $\Derilis[n_2] {\Gamma_2, x \col \tempty} {M_2} B$,
      $\Derilis[n_3] {\Gamma_3} N {A_1}$,
      $n = n_1 + n_2 + n_3$, $\Gamma = (\Gamma_1 \cap \Gamma_2) \cap \Gamma_3$,
      $U = A_1$ or $U = \tempty$.
      Hence $\Derilis[n_1 + n_3 + n_2] {(\Gamma_1 \cap \Gamma_3) \cap \Gamma_2}
      {\Subst {M_1} x N {M_2}} F$.

    \item $\Subst {(M_1 M_2)} x N \Rew{S} M_1 \Subst {M_2} x N$ with
      $x \notin fv(M_1)$, $x \in fv(M_2)$ and $A = F$:
      Then there exist $\Gamma_1$, $\Gamma_2$, $\Gamma_3$, $A_1$, $B$, $n_1$, $n_2$
      and $n_3$ such that $\Derilis[n_1] {\Gamma_1, x \col \tempty} {M_1}
      {B \rightarrow F}$, $\Derilis[n_2] {\Gamma_2, x \col A_1} {M_2} F$,
      $\Derilis[n_3] {\Gamma_3} N {A_1}$, $n = n_1 + n_2 + n_3$,
      $\Gamma = (\Gamma_1 \cap \Gamma_2) \cap \Gamma_3$,
      Hence $\Derilis[n_1 + n_2 + n_3] {\Gamma_1 \cap (\Gamma_2 \cap \Gamma_3)}
      {M_1 \Subst {M_2} x N} F$.

		\item For $\Subst {\Subst M x {N_1}} y {N_2} \equiv
			\Subst {\Subst M y {N_2}} x {N_1}$ with $x \neq y$,
			$x \notin fv(N_2)$, $y \notin fv(N_1)$ and $A = F$:
		There exist $\Gamma_1$, $\Gamma_2$, $n_1$, $n_2$, $U$ and $B$ such that
		$\Gamma = \Gamma_1 \cap \Gamma_2$, $n = n_1 + n_2$, $U = B$ or
		$U = \tempty$,
$\Deri[n_1] {\Gamma_1, y \col U} {\Subst M x {N_1}} F$ and
$\Deri[n_2] {\Gamma_2} {N_2} B$.
Therefore, there exist $\Gamma_3$, $\Gamma_4$, $n_3$, $n_4$, $V$ and $C$ such
that $\Gamma_1, y \col U = \Gamma_3 \cap \Gamma_4$, $n_1 = n_3 + n_4$,
$V = C$ or $V = \tempty$, $\Deri[n_3] {\Gamma_3, x \col V} M F$ and
$\Deri[n_4] {\Gamma_4} {N_1} C$.
By the fact that $y \notin fv(N_1)$ and by Lemma~\ref{lem:basicprop}.2, we have
$\Gamma_4 = \Gamma_4, x \col \tempty$.
Hence, there exists $\Gamma_5$ such that $\Gamma_3 = \Gamma_5, x \col U$ and
$\Gamma_1 = \Gamma_5 \cap \Gamma_4$.
So, $\Deri[n_3] {\Gamma_5, y \col U, x \col V} M F$.
Hence, $\Deri[n_3 + n_2] {(\Gamma_5, x \col V) \cap \Gamma_2} {\Subst M y {N_2}}
F$.
By the fact that $x \notin fv(N_2)$ and by Lemma~\ref{lem:basicprop},
$\Gamma_2 = \Gamma_2, x \col \tempty$.
Hence, $(\Gamma_5, x \col V) \cap \Gamma_2 = \Gamma_5 \cap \Gamma_2, x \col V$.
Therefore, $\Deri[n_3 + n_2 + n_4] {(\Gamma_5 \cap \Gamma_2) \Gamma_4}
{\Subst {\Subst M y {N_2}} x {N_1}} F$ with
$n_3 + n_2 + n_4 = n$ and $(\Gamma_5 \cap \Gamma_2) \cap \Gamma_4 \approx \Gamma$.

    \item The other rules follow the same patterns, especially for the
      propagation of an explicit substitution over another explicit substitution.
      Now concerning the congruent closure of the rules, all cases are straightforward
      but for the following one:

    \item $\Subst M x N \Rew{W} \Subst {M'} x N$ with $M \Rew{W} M'$,
      $x \in fv(M)$, $x \in fv(M')$ and $A = F$:
      Then there exist $\Gamma_1$, $\Gamma_2$, $B$, $n_1$, $n_2$ such that:
      $\Derilis[n_1] {\Gamma_1, x \col B} M F$ and
      $\Derilis[n_2] {\Gamma_2} N B$, $n = n_1 + n_2$ and
      $\Gamma = \Gamma_1 \cap \Gamma_2$.
      By induction hypothesis, there exist $\Gamma_1'$, $C$ and $n_1'$ such that
      $\Gamma_1 \subseteq \Gamma_1'$, $B \subseteq C$, $n_1' \leq n_1$ and
      $\Derilis[n_1'] {\Gamma_1', x \col C} {M'} A$.
      Then there exist $\Gamma_2'$ and $n_2'$ such that
      $\Gamma_2 \subseteq \Gamma_2'$, $n_2' \leq n_2$ and
      $\Derilis[n_2'] {\Gamma_2'} N C$.
      Hence $\Derilis[n_1' + n_2'] {\Gamma_1' \cap \Gamma_2'} {\Subst {M'} x N} F$
      with $n_1' + n_2' \leq n$ and $\Gamma \subseteq \Gamma_1' \cap \Gamma_2'$.\qedhere
    \end{iteMize}
  \end{proof}
\end{toappendix}

\begin{theorem}[Soundness for $\LambdaS$]
  \label{th:LSsoundness}
  If $M$ is a $\LambdaS$-term and $\Derilis \Gamma M A$, then $M \in \SNLambdaS$.
\end{theorem}

\begin{proof}
  We have $\Derilis[n] \Gamma M A$ for some $n$, and strong normalisation is provided by a lexicogrphic argument as follows:
  \begin{iteMize}{$\bullet$}
  \item $\Rew{B}$ strictly decreases $n$
  \item $\Rew{S}$ and $\Rew{W}$ decrease $n$ or do not change it
  \item $\Rew{S, W}$ terminates on its own.\qedhere
  \end{iteMize}
\end{proof}


\subsection{$\LambdaLxr$}

Remember that $\LambdaLxr$-terms are terms that are linear.
In particular the typing rules of abstraction, explicit substitution,
weakening and contraction, degenerate into the following rules when linearity is assumed:

\[\begin{array}{c}

  \infer{\Derili \Gamma {\l x . M} {A \rightarrow F}}
  {A \subseteq C \quad \Derili {\Gamma, x \col C} M F} \qquad

  \infer{\Derili {\Gamma \cap \Delta} {\Subst M x N} F}
  {\Derili \Gamma N B \quad
    \Derili {\Delta, x \col B} M F} \\\\

  \infer{\Derili {\Gamma, x \col A} {\weakening x M} F}
  {\Derili \Gamma M F} \qquad

  \infer{\Derili {\Gamma, x \col A \cap B} {\contraction x y z M} F}
  {\Derili {\Gamma, y \col A, z \col B} M F}

\end{array}
\]

\begin{remark}

Had we defined the $(Weakening)$ and $(Contraction)$ rules as above from the start (\cf Fig.~\ref{fig:InterLambdaSub}), we would have had to add extra conditions for Theorems~\ref{th:SemTe}.5 and \ref{th:SemTe}.6 to hold. This would amount to assuming some of the consequences of linearity. We prefer to keep Theorem~\ref{th:SemTe} and the whole of Section~\ref{sec:semantics} calculus-independent, by giving a more general definition of the two rules.

\end{remark}

\begin{definition}[Reduction in $\LambdaLxr$]\strut

  The reduction and equivalence rules of $\LambdaLxr$ are presented in Fig.~\ref{fig:LambdaLxr_reductions}.

  For a set of rules $E$ from Fig.~\ref{fig:LambdaLxr_reductions}
denotes the congruent closure of the rules in $E$ modulo
  the equivalence rules.

  $\SNLambdaLxr$ denotes the set of strongly normalising $\LambdaLxr$-terms for
  the entire reduction relation.

\end{definition}

\begin{figure}[!h]

  \[
  \begin{array}{|llcll|}

    \upline

    B: & (\l x . M) N & \Rew{} & \Subst M x N & \\

    SR: & \Subst x x N & \Rew{} & N & \\

    SP: & \Subst {(M_1 M_2)} x N & \Rew{} & (\Subst {M_1} x N) M_2 &
    x \in fv(M_1) \\

    & \Subst {(M_1 M_2)} x N & \Rew{} & M_1 (\Subst {M_2} x N) & x \in fv(M_2) \\

    & \Subst {(\l y . M)} x N & \Rew{} & \l y . \Subst M x N &
    x \neq y, y \notin fv(N) \\

    & \Subst {\weakening y M} x N & \Rew{} & \weakening y {\Subst M x N} &
    x \neq y \\

    & \Subst {\contraction y {z1} {z2} M} x N & \Rew{} &
    \contraction y {z1} {z2} {\Subst M x N} &
    y \neq x, y \notin fv(N), z_i \notin fv(N) \\

    W: & \Subst {\weakening x M} x N & \Rew{} & \weakening {fv(N)} M & \\

    D: & \Subst {\contraction x y z M} x N & \Rew{} &
    \contraction X Y Z {\Subst {\Subst {M} y {N_1}} z {N_2}} & \\

    ACC: & \contraction w x v {\contraction  x z y M} & \equiv &
    \contraction w x y {\contraction x z v M} & x \neq y, v \\

    & \contraction x y z M & \equiv & \contraction x z y M & \\

    & \contraction {x'} {y'} {z'} {\contraction x y z M} & \equiv &
    \contraction x y z {\contraction {x'} {y'} {z'} M} &
    x \neq y', z' \& x' \neq y, z \\

    & \Subst {\Subst {M_1} y {M_2}} x N & \Rew{} &
    \Subst {M_1} y {\Subst {M_2} x N} & x \in fv(M_2) \\

    ACW: & \weakening x {\weakening y M} & \equiv &
    \weakening y {\weakening x M} & \\

    CS: & \Subst {\Subst M x {N_1}} y {N_2} & \equiv &
    \Subst {\Subst M y {N_2}} x {N_1} &
    y \notin fv(N_1), x \notin fv(N_2) \\

    & \Subst {\contraction w y z M} x N & \equiv &
    \contraction w y z {\Subst M x N} &
    x \neq w, \& y, z \notin fv(N) \\

    WAbs: & \l x . \weakening y M & \Rew{} & \weakening y {\l x . M} &
    x \neq y \\

    WApp1: & \weakening y M N & \Rew{} & \weakening y {M N} & \\

    WApp2: & M \weakening y N & \Rew{} & \weakening y {M N} & \\

    WSubs: & \Subst M x {\weakening y N} & \Rew{} &
    \weakening y {\Subst M x N} & \\

    Merge: & \contraction w y z {\weakening y M} & \Rew{} &
    R_w^z(M) & \\

    Cross: & \contraction w y z  {\weakening x M} & \Rew{} &\weakening x
    {\contraction w y z M} & x \neq y, x \neq z \\

    CAbs: & \contraction w y z {\l x . M} & \Rew{} &
    \l x . \contraction w y z M &\\

    CApp1: & \contraction w y z {M N} & \Rew{} &
    \contraction w y z M N & y, z \in fv(M) \\

    CApp2: & \contraction w y z {M N} & \Rew{} &
    M \contraction w y z N & y, z \in fv(N) \\

    CSubs: & \contraction w y z {\Subst M x N} & \Rew{} &
    \Subst M x {\contraction w y z N} & y, z \in fv(N)

    \downline 

  \end{array}
  \]
In rule $D$,  $X = fv(N)$, $N_1 = \subst N {\vec X} {\vec Y}$,
  $N_2 = \subst N {\vec X} {\vec Z}$, $Y$, $Z$ fresh sets of variables in bijection with $X$.
In rule $Merge$, $R_w^{z}(M)$ is the renaming of $z$ by $w$ in $M$.
  \caption{Reduction and equivalence rules of $\LambdaLxr$}
  \label{fig:LambdaLxr_reductions}
\end{figure}

\begin{toappendix}

\appendixbeyond 0

\begin{theorem}[Subject Reduction for $\LambdaLxr$]
  If $\Derilis[n] \Gamma M A$ then:
  \begin{iteMize}{$\bullet$}
  \item If $M \Rew{B} M'$, then there exist $\Gamma'$ and $m$ such that
    $\Gamma \subseteq \Gamma'$, $m < n$,  and $\Derilis[m] {\Gamma'} {M'} A$
  \item If $M \Rew{E} M'$ and $B \notin E$, then there exist $\Gamma'$ and $m$ such that
    $\Gamma \subseteq \Gamma'$, $m \leq n$ and $\Derilis[m] {\Gamma'} {M'} A$.
\item If $M \equiv M'$ then there exist $\Gamma'$ such that
$\Gamma \approx \Gamma'$ and $\Deri[n] {\Gamma'} {M'}  A$.
  \end{iteMize}
\end{theorem}

\end{toappendix}

\begin{toappendix}[
\begin{proof}
See Appendix~\thisappendix.
\end{proof}
]

\begin{proof} 
First by induction on $M \Rew{B} M'$ (resp. $M \Rew{E} M'$, $M \equiv M'$) then
by induction on $A$.
The proof is similar to the proof of Subject reduction.
Here we will only detail the cases of rules $D$, $W$, $B$ and $Merge$:
\begin{iteMize}{$\bullet$}

\item For $\Subst {\contraction x y z M} x N \Rew{}
\contraction X Y Z {\Subst {\Subst { M} y {N_1}} z {N_2}}$ with $A = F$:
There exist $n_1$, $n_2$, $n_3$, $B_1$, $C_1$, $\Delta$, $\vec{B}$ and $\vec{C}$
such that
$\Gamma = \Delta, \vec{X} \col \vec{B} \cap \vec{C}$, $n = n_1 + n_2 + n_3$,
$\Deri[n_1] {\Delta, y \col B_1, z \col C_1} M F$,
$\Deri[n_2] {\vec{X} \col \vec{B}} N {B_1}$ and
$\Deri[n_3] {\vec{X} \col \vec{C}} N {C_1}$.
Hence, $\Deri[n_2] {\vec{Y} \col \vec{B}} {N_1} {B_1}$ and
$\Deri[n_3] {\vec{Z} \col \vec{C}} {N_2} {C_1}$.
Therefore, $\Deri[n] {\Delta, \vec{Y} \col \vec{B}, \vec{Z} \col \vec{C}}
{\Subst {\Subst M y {N_1}} z {N_2}} F$. 
Then we can conclude.

\item For $\Subst {\weakening x M} x N \Rew {} \weakening {fv(N)} M$ with
$X = fv(N)$ and $A = F$:
Then, there exist $n_1$, $n_2$, $\Delta$, $\vec{B}$ and $B_1$ such that
$n = n_1 + n_2$, $\Gamma = \Delta, \vec{X} \col \vec{B}$,
$\Deri[n_1] {\Delta} M F$, $x \notin \dom \Delta$,
and $\Deri[n_2] {\vec{X} \col \vec{B}} N {B_1}$.
Hence, for all $y \in X$, $y \notin \dom \Delta$.
Then we can conclude.

\item For $(\l x . M) N \Rew{B} \Subst M x N$, the proof is exactly as the proof
for the rule $B$ in $\LambdaS$.

\item For $\contraction w y z {\weakening y M} \Rew{} R_w^{z}(M)$ with $A = F$:
Then, there exist $\Gamma_1$, $A$ and $B$ such that
$\Gamma = \Gamma_1, w \col A \cap B$ and
$\Deri[n] {\Gamma_1, y \col A, z \col B} {\weakening y M} F$.
Hence, $\Deri[n] {\Gamma_1, z \col B} M F$.
Therefore, $\Deri[n] {\Gamma_1, w \col B} {R_w^{z}(M)} F$ with
$\Gamma \subseteq (\Gamma_1, w \col B)$.

\item For $\contraction w x v {\contraction x y z M} \equiv
\contraction w x y {\contraction x z v M}$ with $F = A$:
Then, there exist $\Gamma_1$, $A_1$, $A_2$, $A_3$, such that
$\Gamma = \Gamma_1, w \col (A_1 \cap A_2) \cap A_3$ and
$\Deri[n] {\Gamma_1, y \col A_1, z \col A_2, v \col A_3} M F$.
Therefore $\Deri[n] {\Gamma_1, w \col (A_2 \cap A_3) \cap A_1}
{\contraction w x y {\contraction x z v M}} F$.\qedhere

\end{iteMize}
\end{proof}

\end{toappendix}

\begin{theorem}[Soundness for $\LambdaLxr$]
  \label{th:Llxrsoundness}
  If $M$ is a $\LambdaLxr$-term and $\Derilis \Gamma M A$, then
  $M \in \SNLambdaLxr$.
\end{theorem}

\begin{proof}
  Similarly to $\LambdaS$ we have $\Derilis[n] \Gamma M A$ for some $n$, and strong normalisation is provided by a lexicogrphic argument as follows:
  \begin{iteMize}{$\bullet$}
  \item $\Rew{B}$ strictly decreases $n$
  \item Any other reduction $\Rew{E}$ decreases $n$ or does not change it
  \item The reduction system without the $B$ rule terminates on its own~\cite{KLIaC06}.\qedhere
\end{iteMize}
\end{proof}


\section{Denotational semantics for strong normalisation}

\label{sec:semantics}


In this section we show how to use non-idempotent intersection types to simplify the methodology of~\cite{CoquandSpiwack07}, which we briefly review here:

The goal is to produce modular proofs of strong normalisation for various {\em source typing systems}.  The problem is reduced to the strong normalisation of a unique {\em target system} of intersection types, chosen once and for all.  This is done by interpreting each term $t$ as the set $\SemTe t{}$ of the intersection types that can be assigned to $t$ in the target system.  Two facts then remain to be proved:

\begin{enumerate}[(1)]
\item if $t$ can be typed in the source system, then $\SemTe t{}$ is not empty
\item the target system is strongly normalising
\end{enumerate}
The first point is the only part that is specific to the source typing system: it amounts to turning the interpretation of terms into a filter model of the source typing system. The second point depends on the chosen target system: as~\cite{CoquandSpiwack07} uses a system of {\em idempotent} intersection types (extending the simply-typed $\l$-calculus), their proof involves the usual reducibility technique~\cite{Girard72,tait75}.  But this is somewhat redundant with point 1 which uses similar techniques to prove the correctness of the filter model with respect to the source system.\footnote{If reducibility techniques are needed for the latter, why not use them on the source system directly (besides formulating a modular methodology)?}

In this paper we propose to use {\em non-idempotent} intersection types for the target system, so that point 2 can be proved with simpler techniques than in~\cite{CoquandSpiwack07} while point 1 is not impacted by the move.  In practice we propose \InterLambdaSub\ as the target system (that of~\cite{bernadetleng11} would work just as well). The present section shows the details of this alternative.

Notice that \InterLambdaSub\ is not an extension of the simply-typed $\l$-calculus, in that a typing tree in the system of simple types is not a valid typing tree in system \InterLambdaSub, which uses non-idempotent intersections (while it is a valid typing tree in the system of~\cite{CoquandSpiwack07} which uses idempotent intersections).
But a nice application of our proposed methodology is that, by taking the simply-typed lambda-calculus as the source system, we can produce a typing tree in  \InterLambdaSub\ from a typing tree with simple types. We do not know of any more direct encoding.

\subsection{I-filters}

\label{sec:filters}

The following filter constructions only involve the syntax of types and are
independent from the chosen target system.

\begin{definition}[I-filter]\strut
\begin{iteMize}{$\bullet$}
\item
  An I-filter is a set $v$ of $A$-types such that:
  \begin{iteMize}{$-$}
  \item for all $A$ and $B$ in $v$ we have $A \cap B \in v$
  \item for all $A$ and $B$, if $A \in v$ and $A \subseteq B$, then $B \in v$
  \end{iteMize}
\item In particular the empty set and the sets of all $A$-types are I-filters
and we write them $\bottom$ and $\top$ respectively.
\item Let $\ValDom$ be the sets of all non-empty I-filters; we call such
I-filters {\em  values}.
\item Let $\ValDomE$ be the sets of all I-filters ($\ValDomE=\ValDom\cup\{\bottom\}$).
\end{iteMize}
\end{definition}

While our intersection types differ from those in~\cite{CoquandSpiwack07} (in 
that idempotency is dropped), the stability of a filter under type intersections
makes it validate idempotency (it contains $A$ \iff\ it contains $A\cap A$,
etc).
This makes our filters very similar to those in~\cite{CoquandSpiwack07}, so we
can plug-in the rest of the methodology with minimal change.

\begin{remark}[Basic properties of I-filters]\strut
\begin{enumerate}[(1)]
\item If $(v_i)_{i \in I}$ is an non empty family of $\ValDomE$, then
$\bigcap_{i \in I} v_i \in \ValDomE$.
\item If $v$ is a set of $A$-types, then there is a smallest $v' \in \ValDomE$
such that $v \subseteq v'$ and we write $\Gen v \eqdef v'$.
\item If $v$ is a set of $F$-types, then $\Gen v$ is the closure of $v$ under
finite intersections.
\item If $v \in \ValDomE$, then $v = \Gen {\{ F \mid F \in v \}}$.
\item If $u$ and $v$ are sets of $F$-types such that $\Gen u = \Gen v$, then
$u = v$.
\end{enumerate}
\end{remark}

Hence, in order to prove that two I-filters are equal we just have to prove that
they contain the same $F$-types. 

I-filters form an applicative structure:
\begin{definition}[Application of I-filters]
If $u$, $v$ are in $\ValDomE$, then define \[u\ap v \eqdef \Gen {\{ F \mid
\exists A \in v, (A \rightarrow F) \in u \}}\]
\end{definition}

\begin{remark}
For all $u \in \ValDomE$, $u\ap \bottom = \bottom\ap u = \bottom$, and for all
$u \in \ValDom$, $\top\ap u = \top$.
\end{remark}


\begin{definition}[Environments and contexts]\strut

An {\em environment} is a map from term variables $x,y,\ldots$ to I-filters.

If $\rho$ is an environment and $\Gamma$ is a context, we say that
$\Gamma \in \rho$, or $\Gamma$ is {\em compatible with} $\rho$, if for all $x$,
$\Gamma(x) = \tempty$ or $\Gamma(x) \in \rho(x)$.

Assume $\rho$ is an environment, $u$ is an I-filter and $x$ is a variable.
Then, the environment $\rho, x \mapsto u$ is defined as follows:
\[\begin{array}{llll}
(\rho, x \mapsto u)(x) & \eqdef & u & \\
(\rho, x \mapsto u)(y) & \eqdef & \rho(y) & \forall y \neq x
\end{array}
\]
\end{definition}








\begin{remark}[Environments are I-filters of contexts]
\footnote{Conversely, if $E$ is an $I$-filter of contexts, then $\rho$, defined
by $\rho(x) = \{\Gamma(x) \neq \tempty \mid \Gamma \in E\}$ for all $x$, is an
environment.}
Let $\rho$ be an environment.
\begin{enumerate}[(1)]
\item
If $\Gamma \in \rho$ and $\Gamma' \in \rho$, then $\Gamma \cap \Gamma' \in \rho$.
\item
If $\Gamma \in\rho$ and $\Gamma'$ is a context such that $\Gamma \subseteq
\Gamma'$, then $\Gamma' \in \rho$.
\end{enumerate}
\end{remark}

\subsection{Semantics of terms as I-filters}

The remaining ingredients now involve the target system; we treat here
\InterLambdaSub.
\begin{definition}[Interpretation of terms]\label{def:intterms}\strut\\
  If $M$ is a term and $\rho$ is an environment we define
  \[\SemTe M \rho \eqdef \{ A \mid \exists \Gamma \in\rho, \Derilis \Gamma M A \}
  \]
\end{definition}

\begin{remark}
  $\SemTe M \rho \in \ValDomE$, and therefore $\SemTe M \rho=\Gen{\{ F \mid
    \exists \Gamma \in\rho, \Derilis \Gamma M F \}}$.
\end{remark}


\begin{toappendix}
  \begin{theorem}[Inductive characterisation of the interpretation]
    \strut\label{th:SemTe}
    \begin{enumerate}[\em(1)]
    \item $\SemTe x \rho = \rho(x)$
    \item $\SemTe {M N} \rho = {\SemTe M \rho}\ap {\SemTe N \rho}$
    \item $\SemTe {\lambda x.M} \rho\ap u = \SemTe {M} {\rho, x \mapsto u}$
      if $u \neq \bot$.
    \item $\SemTe {\Subst M x N} \rho = \SemTe M {\rho, x \mapsto \SemTe N \rho}$
      if $\SemTe N \rho \neq \bot$
    \item $\SemTe {\weakening x M} \rho = \SemTe M \rho$ if $\rho(x) \neq \bot$.
    \item $\SemTe {\contraction x y z M} \rho = \SemTe M {\rho, y \mapsto \rho(x),
        z \mapsto \rho(x)}$.

    \end{enumerate}
  \end{theorem}
\end{toappendix}

\begin{toappendix}
  [\begin{proof}See Appendix~\thisappendix.\end{proof}]
  \begin{proof}
    To prove equalities between I-filters, we only have to prove that they have
    the same $F$-types.
    \begin{enumerate}[(1)]
    \item If $F \in \SemTe x \rho$, then there exists $\Gamma \in\rho$ such that
      $\Derilis \Gamma x F$, so $\Gamma = (x \col F)$ so $F \in \rho(x)$.
      Conversely, if
      $F \in \rho(x)$, then $(x \col F) \in \rho$ and $\Derilis {x \col F} x F$. So
      $\SemTe x \rho = \rho(x)$.
    \item Let $F \in \SemTe{M N} \rho$. There exists $\Gamma \in \rho$ such that
      $\Derilis \Gamma {M N} F$. Hence, there exist $\Gamma_1$, $\Gamma_2$ and $A$
      such that $\Derilis {\Gamma_1} M {A \rightarrow F}$ and
      $\Derilis {\Gamma_2} N A$ and
      $\Gamma = \Gamma_1 \cap \Gamma_2$. So $\Gamma \subseteq \Gamma_1$, and
      $\Gamma_1 \in \rho$. Hence, $A \rightarrow F \in \SemTe M \rho$. We also have
      $A \in \SemTe N \rho$. So we have $F \in \SemTe M \rho \ap \SemTe N \rho$.

      Conversely, let $F \in \SemTe M \rho \ap \SemTe N \rho$.
      There exists $A$ such that $A \rightarrow F \in \SemTe M \rho$ and
      $A \in \SemTe N \rho$.
      So there exists $\Gamma \in \rho$ such that
      $\Derilis \Gamma M {A \rightarrow F}$ and there exists
      $\Delta\in\rho$ such that $\Derilis \Delta N A$.
      Hence, $\Gamma \cap \Delta \in \rho$ and
      $\Derilis {\Gamma \cap \Delta} {M N} F$.
      So $F \in \SemTe {M N} \rho$.

    \item Let $F \in \SemTe {\l x . M} \rho \ap u$.
      There exists $A \in u$ such that $A \rightarrow F \in \SemTe {\l x . M} \rho$.
      So there exists $\Gamma \in \rho$ such that $\Derilis \Gamma {\l x . M} {A \rightarrow F}$.
      Hence, there exists $U$ such that $A \subseteq U$ and
      $\Derilis {\Gamma, x \col U} M F$.

      \begin{iteMize}{$\bullet$}
      \item If $U = \tempty$, then $(\Gamma, x \col U) = \Gamma \in \rho$.
        Since $\rho \subseteq (\rho, x \mapsto u)$, we have
        $(\Gamma, x \col U) \in (\rho, x \mapsto u)$.
      \item If not, then $U \in u$, and we have
        $(\Gamma, x \col U) \in (\rho, x \mapsto u)$.
      \end{iteMize}
      So we have $F \in \SemTe M {\rho, x \mapsto u}$.

      Conversely, let $F \in \SemTe M {\rho, x \mapsto u}$. There
      exists $\Gamma \in (\rho, x \mapsto u)$ such that $\Derilis \Gamma M F$.
      \begin{iteMize}{$\bullet$}
      \item If $x \in \FV M$, then there exist
        $\Gamma_1 \in \rho$ and $A \in u$ such that $\Gamma = \Gamma_1, x \col A$
        (using Lemma~\ref{lem:basicprop}.1).
        So we have $\Derilis {\Gamma_1} {\l x . M} {A \rightarrow F}$, and then
        $A \rightarrow F \in \SemTe {\l x . M} \rho$. Hence, we have
        $F \in \SemTe {\l x . M} \rho \ap u$.
      \item If $x \notin \FV M$, then for all $y \in \dom\Gamma$, $y \neq x$
        (using Lemma~\ref{lem:basicprop}.1).
        So $\Gamma \in \rho$ and $\Derilis {\Gamma, x \col \tempty} M F$.
        Since $u$ is a value, there exist $A \in u$ and
        $\Derilis \Gamma {\l x . M} {A \rightarrow F}$.
        Hence, $A \rightarrow F \in \SemTe {\l x . M} \rho$ and finally
        $F \in \SemTe {\l x . M} \rho \ap u$.
      \end{iteMize}

    \item Let $F \in \SemTe {\Subst M x N} \rho$.
      So there exists $\Gamma \in \rho$ such that $\Derilis \Gamma {\Subst M x N} F$.
      Hence there exist $\Gamma_1$, $\Gamma_2$, $A$ and $U$ such that
      $\Gamma = \Gamma_1 \cap \Gamma_2$, $\Derilis {\Gamma_1} N A$,
      $\Derilis {\Gamma_2, x \col U} M F$ and $U = A$ or $U = \tempty$.
      Hence $\Gamma_1 \in \rho$ and then $A \in \SemTe N \rho$.
      Therefore because $\Gamma_2 \in \rho$ we also have
      $(\Gamma_2, x \col A) \in (\rho, x \mapsto \SemTe N \rho)$
      and $(\Gamma_2, x \col \tempty) \in (\rho, x \mapsto \SemTe N \rho)$.
      Hence $(\Gamma_2, x \col U) \in (\rho, x \mapsto \SemTe N \rho)$.
      So we have $F \in \SemTe M {\rho, x \mapsto \SemTe N \rho}$.

      Conversely, if $F \in \SemTe M {\rho, x \mapsto \SemTe N \rho}$ then there
      exists $\Gamma \in (\rho, x \mapsto \SemTe N \rho)$ such that
      $\Derilis \Gamma M F$. Hence there exist $\Gamma'$ and $U$ such that
      $\Gamma' \in \rho$ and $U \in \SemTe N \rho$ or $U = \tempty$.
      \begin{iteMize}{$bullet$}

      \item If $U \in \SemTe N \rho$ there exist $A$ and $\Delta$ such that
        $\Derilis \Delta N A$ and $A = U$.

      \item If $U = \tempty$, then because $\SemTe N \rho \neq \bot$ there exists
        $A \in \SemTe N \rho$. Therefore there exist $\Delta \in \rho$ such that
        $\Derilis \Delta N A$.

      \end{iteMize}
      Hence $\Derilis {\Delta \cap \Gamma'} {\Subst M x N} F$ with
      $(\Delta \cap \Gamma') \in \rho$.
      Therefore $F \in \SemTe {\Subst M x N} \rho$.

    \item Let $F \in \SemTe {\weakening x M} \rho$.
      So there exists $\Gamma \in \rho$ such that
      $\Derilis \Gamma {\weakening x M} F$.
      So there exist $\Gamma'$, $U$ and $A$ such that
      $\Gamma = (\Gamma', x \col U \cap A)$ and
      $\Derilis {\Gamma', x \col U} M F$.
      Therefore, $\Gamma \subseteq (\Gamma', x \col U)$.
      Hence $(\Gamma', x \col U) \in \rho$.
      Therefore $F \in \SemTe M \rho$.

      Conversely, if $F \in \SemTe M \rho$, then there exists $\Gamma \in \rho$ such
      that $\Derilis \Gamma M F$. $\rho(x) \neq \bot$, so there exists
      $A \in \rho(x)$.
      Also, there exist $\Gamma'$ and $U$ such that
      $\Gamma = (\Gamma', x \col U)$ and $U = \tempty$ or $U \in \rho(x)$.
      So we have $U \cap A \in \rho(x)$ and $\Gamma' \in \rho$.
      Hence $(\Gamma, x \col U \cap A) \in \rho$ and
      $\Derilis {\Gamma, x \col U \cap A} {\weakening x M} F$.
      Therefore $F \in \SemTe {\weakening x M} \rho$.

    \item Let $F \in \SemTe {\contraction x y z M} \rho$.
      So there exists $\Gamma \in \rho$ such that
      $\Derilis \Gamma {\contraction x y z M} F$.
      Hence there exist $\Gamma'$, $U$, $V_1$ and $V_2$ such that
      $\Gamma = (\Gamma', x \col U \cap (V_1 \cap V_2))$ and
      $\Derilis {\Gamma', x \col U, y \col V_1, z \col V_2} M F$.
      Therefore $\Gamma' \in \rho$ and $U \cap (V_1 \cap V_2) \in \rho(x)$
      or $U \cap (V_1 \cap V_2) = \tempty$.
      So $U$, $V_1$ and $V_2$ are either equal to $\tempty$ or are in $\rho(x)$.
      Hence $(\Gamma', x \col U, y \col V_1, z \col V_2) \in
      (\rho, y \mapsto \rho(x), z \mapsto \rho(x))$.
      Therefore $F \in \SemTe M {\rho, y \mapsto \rho(x), z \mapsto \rho(x)}$.

      Conversely, if $F \in \SemTe M {\rho, y \mapsto \rho(x), z \mapsto \rho(x)}$
      , then there exists $\Gamma \in (\rho, x \mapsto \rho(x), y \mapsto \rho(y))$.
      So there exist $\Gamma'$, $U$, $V_1$ and $V_2$ such that
      $\Gamma = (\Gamma', x \col U, y \col V_1, z \col V_2)$ and $U$, $V_1$ and
      $V_2$ are either equal to $\tempty$ or are in $\rho(x)$.
      Hence $(U \cap (V_1 \cap V_2)) \in \rho(x)$ or is equal to $\tempty$.
      So we have $\Derilis {\Gamma', x \col U \cap (V_1 \cap V_2)}
      {\contraction x y z M} F$ with
      $(\Gamma', x \col U \cap (V_1 \cap V_2)) \in \rho$.
      Therefore $F \in \SemTe {\contraction x y z M} \rho$.

    \end{enumerate}
  \end{proof}
\end{toappendix}

This theorem makes \InterLambdaSub\ a suitable alternative as a target system: the filter models of the source systems treated in~\cite{CoquandSpiwack07} can be done with a system of non-idempotent intersection types.  While we could develop those constructions, we prefer to cover a new range of source systems: those with second-order quantifiers such as System $F$.

\subsection{An example: System F and the likes}






\begin{definition}[Types and Typing System]  Types are built by the following grammar:
  \[\gA,\gB,\ldots\recdef \alpha\mid \gA\arr \gB \mid \gA\cap \gB \mid \forall
  \alpha \gA\]
  where $\alpha$ denotes a type variable, $\forall \alpha \gA$ binds $\alpha$ in
  $\gA$, types are considered modulo $\alpha$-conversion, and $\TV\gA$ denotes the
  free (type) variables of $\gA$.

  Typing contexts, denoted $\gGamma$, $\gDelta,\ldots$ are partial maps from term
  variables to types, and $(x\col \gA)$ denotes the map from $x$ to $\gA$.

  Let $\SysS$ be the typing system consisting of the rules in Fig.~\ref{fig:TypMix}.

  Typability in system $\SysS$ will be expressed by judgements of the form $\DeriliS \gGamma M \gA$.

\end{definition}

\begin{figure}[!h]
  \[
  \begin{array}{|c|}
    \upline
    \infer{\Derili{\gGamma,x\col \gA} x \gA}{\strut}
    \qquad
    \infer{\Derili{\gGamma} {\lambda x.M} {\gA\arr \gB}}
    {\Derili{\gGamma,x\col \gA} M \gB}
    \qquad
    \infer{\Derili{\gGamma} {M\ N} \gB}{\Derili\gGamma M {\gA\arr \gB}\quad
      \Derili\gGamma N \gA}
    \\\\
    \infer{\Derili \gGamma M {\gA\cap \gB}}{\Derili \gGamma M {\gA}\quad
      \Derili \gGamma M {\gB}}
    \qquad
    \infer{\Derili \gGamma M {\gA}}{\Derili \gGamma M {\gA\cap \gB}}
    \qquad
    \infer{\Derili \gGamma M \gB}{\Derili \gGamma M {\gA\cap \gB}}
    \\\\
    \infer[\alpha\notin\TV\gGamma]{\Derili \gGamma M {\forall \alpha \gA}}
    {\Derili \gGamma M {\gA}}
    \qquad
    \infer{\Derili \gGamma M {\subst \gA \alpha \gB}}
    {\Derili \gGamma M {\forall \alpha \gA}}
    \\\\
    \infer{\Derili \gGamma {\Subst M x N} \gB}{\Derili \gGamma N \gA \quad
      \Derili {\gGamma, x \col \gA} M \gB}
    \qquad
    \infer{\Derili {\gGamma, x \col \gA} {\weakening x M} \gB}
    {\Derili \gGamma M \gB \quad x \notin dom(\gGamma)}
    \qquad
    \infer{\Derili {\gGamma, x \col \gA} {\contraction x y z M} \gB}
    {\Derili {\gGamma, y \col \gA, z \col \gA} M \gB}
    \downline
  \end{array}
  \]
  \caption{Miscellaneous Typing Rules}
  \label{fig:TypMix}
\end{figure}

\begin{remark}
  Note that the typing rules in System $\SysS$ do not necessarily follow the philosophy of the $\LambdaLxr$-calculus and the $\LambdaS$-calculus.
  For example, we would expect a typing system for $\LambdaS$ or $\LambdaLxr$ to be such that the domain of the typing context is exactly the set of free variables in the typed term (this leads to interesting properties and better encodings into proof-nets -see \eg\cite{KLIaC06}).
  
  However here, we are only interested in strong normalisation, and we therefore consider a typing system as general as possible (hence the accumulation of rules in Fig.~\ref{fig:TypMix}), \ie a typing system such that the terms that are typed in an appropriate typing system (such as that of~\cite{KLIaC06} for $\LambdaLxr$) can be typed here. This is the case of System $\SysS$. Alternatively, we could also adapt and tailor the proof of strong normalisation below to the specific typing system in which we are interested.
\end{remark}

\subsection{An intuitionistic realisability model}
\label{sec:ifiltermodel}

We now build the model \Mfiltersi\ as follows:
\begin{definition}[Realisability Predicate]
  A {\em realisability predicate} is a subset $X$ of $\ValDom$ containing $\top$. We define $\TP\ValDom$ as the set of realisability predicates.
\end{definition}

\begin{lemma}[Shape of realisability predicates]\strut
  \begin{enumerate}[\em(1)]
  \item If $(X_i)_{i \in I}$ is a non empty family of $\TP\ValDom$, then
    $\bigcap_{i \in I} X_i \in \TP\ValDom$.
  \item If $X$ and $Y$ in $\TP\ValDom$, then $X \rightarrow Y \in \TP\ValDom$ where
    $X \rightarrow Y$ is defined as
    \[X \rightarrow Y \eqdef \{ u \mid \forall v \in X, u\ap v \in Y \}\]
  \end{enumerate}
\end{lemma}
\begin{proof}
  The only subtle point is the second one: First, for all $v \in X$, $v \neq \bot$
  and thus $\top\ap v = \top \in Y$. So $\top \in X \rightarrow Y$.
  Second, suppose that $\bot \in X \rightarrow Y$. As $X\neq \emptyset$, there is 
  $u\in X$, for which $\bot\ap u = \bot \in Y$, which contradicts the fact that
  $Y\in \TP\ValDom$.
\end{proof}

We can now interpret types:
\begin{definition}[Interpretation of types]\strut\\
  {\em Valuations} are mappings from type variables to elements of $\TP\ValDom$.

  Given such a valuation $\sigma$, the interpretation of types is defined as
  follows:
  \[\begin{array}{c}
    \begin{array}{lll}
      \SemTyP \alpha \sigma &\eqdef \sigma(\alpha)\\
      \SemTyP {\gA \arr \gB} \sigma &\eqdef {\SemTyP \gA \sigma}\arr{\SemTyP \gB
        \sigma}
    \end{array}
    \qquad
    \begin{array}{lll}
      \SemTyP {\gA \cap \gB} \sigma &\eqdef \SemTyP \gA \sigma \cap \SemTyP \gB
      \sigma\\
      \SemTyP {\forall \alpha \gA}\sigma &\eqdef \bigcap_{X\in \TP\ValDom} \SemTyP \gA
      {\sigma, \alpha \mapsto X}
    \end{array}
  \end{array}
  \]
  The interpretation of typing contexts is defined as follows:
  \[
  \SemTyP \gGamma\sigma\eqdef \{ \rho\mid \forall (x\col \gA)\in\gGamma,\rho(x)\in
  \SemTyP \gA\sigma\}
  \]
\end{definition}
Finally we get Adequacy:
\begin{lemma}[Adequacy Lemma]
  If $\DeriliS \gGamma M \gA$, then for all valuations $\sigma$ and for all
  mappings $\rho\in\SemTyP\gGamma\sigma$ we have $\SemTe M\rho\in\SemTyP
  \gA\sigma$.
\end{lemma}
\begin{proof}
  By induction on the derivation of $\DeriliS \gGamma M \gA$, using
  Theorem~\ref{th:SemTe} (and the fact that $\SemTyP {\subst \gA\alpha\gB}
  \sigma=\SemTyP {\gA} {\sigma,\alpha\mapsto \SemTyP {\gA} \sigma}$, which is
  proved by induction on $\gA$).
\end{proof}

\begin{corollary}[Strong normalisation of $\SysS$]
  If $\DeriliS \gGamma M \gA$, then  $M\in\SN{}$.
\end{corollary}
\begin{proof}
  Applying the previous lemma with $\sigma$ mapping every type variable to
  $\{\top\}$ and $\rho$ mapping all term variables to $\top$, we get
  $\SemTe M\rho\in\SemTyP \gA\sigma$, so $\SemTe M\rho\neq\bottom$.
  Hence, $M$ can be typed in \InterLambdaSub, so $M\in\SN{}$. 
\end{proof}
The advantage of non-idempotent intersection types (over idempotent ones) lies
in the very last step of the above proof:
here the typing trees of \InterLambdaSub\ get smaller with every
$\beta$-reduction (proof of Theorem~\ref{th:soundness}), while a reducibility
technique as in~\cite{CoquandSpiwack07} combines yet again an induction on types
with an induction on typing trees similar to that in the Adequacy Lemma.

\subsection{Orthogonality models}

In this section we show how the above methodology can be integrated to the theory of {\em orthogonality}, \ie how this kind of filter model construction can be captured by orthogonality techniques~\cite{girard-ll,DanosKrivine00, Krivine01,MunchCSL09}.  These techniques are particularly suitable to prove that typed terms satisfy some property~\cite{Parigot97,mellies05recursive,LM:APAL07}, the most well-known of which being Strong Normalisation.

For this we define an abstract notion of orthogonality model for the system $\SysS$ defined in Fig.~\ref{fig:TypMix}.  In particular our definition also applies to sub-systems such as the simply-typed $\l$-calculus, the idempotent intersection type system, System~$F$, etc. We could also adapt it with no difficulty to accommodate System~$F_\omega$.

Orthogonality techniques and the filter model construction from Section~\ref{sec:filters} (with the sets $\ValDom$ and $\ValDomE$) inspire the notion of orthogonality model below. First we need the following notations:

\begin{notation}
Given a set $\mathcal D$, let $\mathcal D^*$ be the set of lists of elements of $\mathcal D$, with $\el$ representing the empty list and $\cons u {\vec v}$ representing the list of head $u$ and tail $\vec{v}$.
\end{notation}

\begin{definition}[Orthogonality model]\strut\\
  \label{prop}An {\em orthogonality model} is a 4-tuple
  $(\ValDomE,\ValDom,\orth{}{},\SemTe {\_ }{\_} )$ where
  \begin{iteMize}{$\bullet$}
  \item $\ValDomE$ is a set, called the {\em support}
  \item $\ValDom\subseteq\ValDomE$ is a set of elements called {\em values}
  \item $\orth{}{}\subseteq\ValDom\times\ValDomL$ is called the
    \emph{orthogonality relation}
  \item $\SemTe {\_ }{\_}$ is a function mapping every term $M$ (typed or untyped)
    to an element $\SemTe {M}{\rho}$ of the support, where $\rho$ is a parameter
    called {\em environment} mapping term variables to values.
  \item the following axioms are satisfied:
    \begin{enumerate}[({A}1)]
    \item[(A1)]
      For all $\rho$, $\vec v$, $x$,\hfill
      if $\orth{\rho(x)} {\vec v}$, then $\orth {\SemTe x\rho} {\vec v}$.
    \item[(A2)]
      For all $\rho$, $\vec v$, $M_1$, $M_2$,\hfill
      if $\orth {\SemTe {M_1}\rho}{(\cons {\SemTe {M_2}\rho}{\vec v})}$, then
      $\orth {\SemTe {M_1\ M_2}\rho}{\vec v}$. 
    \item[(A3)]
      For all $\rho$, $\vec v$, $x$, $M$ and for all values $u$,\\\strut\hfill
      if $\orth {\SemTe {M}{\rho,x\mapsto u}}{\vec v}$, then
      $\orth {\SemTe {\lambda x.M}\rho}{(\cons u{\vec v})}$. 
    \item[(A4)]
      For all $\rho$, $\vec v$, $x$, $M_1$, $M_2$,\\\strut\hfill if 
      $\SemTe {M_2}\rho$ is a value and
      $\orth {\SemTe {M_1}{\rho,x\mapsto \SemTe {M_2}\rho}}
      {\vec v}$, then $\orth {\SemTe {\Subst {M_1}x{M_2}}{\rho}}{\vec v}$. 
    \item[(A5)]
      For all $\rho$, $\vec v$, $x$, $M$ (such that $x\notin fv(M)$) and for all values $u$,\\\strut\hfill    
      if $\orth {\SemTe {M}\rho}{\vec v}$, then $\orth {\SemTe {\weakening x M}
        {\rho,x\mapsto u}}{\vec v}$. 
    \item[(A6)]
      For all $\rho$, $\vec v$, $x,y,z$ (distinct variables), $M$ (such that $x\notin fv(M)$) and for all values $u$,\\\strut\hfill    
      if $\orth {\SemTe {M}{\rho,y\mapsto u,z\mapsto u}}{\vec v}$, then
      $\orth {\SemTe {\contraction x y z M}{\rho,x\mapsto u}}{\vec v}$. 
    \end{enumerate}
  \end{iteMize}
\end{definition}

In fact, $\ValDom$ and $\orth{}{}$ are already sufficient to interpret any type $A$ as a set $\SemTyP{A}{}$ of values: if types are seen as logical formulae, we can see this construction as a way of building some of their realisability / set-theoretical models.

There is no notion of computation pertaining to values, but the interplay between the interpretation of terms and the orthogonality relation is imposed by the axioms so that the Adequacy Lemma (which relates typing to semantics) holds:\\
\strut\hfill
If $\Deri[\SysS] {} M \gA$, then $\SemTe M{}\in\SemTyP \gA{}$
\hfill\strut

\subsubsection{Semantics of types and Adequacy Lemma}

\begin{definition}[Orthogonal]\strut
  \begin{iteMize}{$\bullet$}
  \item
    If $X\subseteq\ValDom$, then let 
    \(\uniorth X \eqdef \{ \vec{v}\in \ValDomL \mid \forall u \in X, \orth u
    {\vec{v}} \}\)
  \item
    If $Y\subseteq\ValDomL$, then let
    \(\uniorth Y \eqdef \{ u\in \ValDom \mid \forall \vec{v} \in Y,
    \orth u {\vec{v}} \}\)
  \end{iteMize}
\end{definition}

\begin{remark}
  If $X\subseteq\ValDom$ or $X\subseteq\ValDomL$, then $X \subseteq \biorth X$
  and $\triorth X= \uniorth X$.
\end{remark}

\begin{definition}[Lists and Cons construct]\strut\\
  If $X\subseteq\ValDom$ and $Y\subseteq\ValDomL$, then define \(\cons X Y
  \eqdef \{\cons u {\vec v}\mid u\in X,\vec v \in Y \}\).
\end{definition}

\begin{definition}[Interpretation of types]\strut\\
  Mappings from type variables to subsets of $\ValDomL$ are called
  {\em valuations}.

  Given such a valuation $\sigma$, the interpretation of types is defined as
  follows:
  \[ \begin{array}{c}
    \begin{array}{lll}
      \SemTyN \alpha \sigma &\eqdef \sigma(\alpha)\\
      \SemTyN {\gA \arr \gB} \sigma &\eqdef \cons{\SemTyP \gA \sigma}{\SemTyN \gB \sigma}\\
    \end{array}
    \qquad
    \begin{array}{lll}
      \SemTyN {\gA \cap \gB} \sigma &\eqdef \SemTyN \gA \sigma \cup \SemTyN \gB
      \sigma\\
      \SemTyN {\forall \alpha \gA}\sigma &\eqdef \bigcup_{Y\subseteq\ValDomL}
      \SemTyN \gA {\sigma, \alpha \mapsto Y}\\
    \end{array}\\[10pt]
    \SemTyP \gA \sigma \eqdef \uniorth {\SemTyN \gA\sigma}
  \end{array}
  \]
  The interpretation of typing contexts is defined as follows:
  \[
  \SemTyP \gGamma\sigma\eqdef \{ \rho\mid \forall (x\col \gA)\in\gGamma,\rho(x)
  \in\SemTyP \gA\sigma\}
  \]
\end{definition}

\begin{remark}
  Note that $\SemTyN{\subst\gA\alpha\gB}\sigma=\SemTyN{\gA}{\sigma,\alpha\mapsto
    \SemTyN{\gB}\sigma}$ and $\SemTyP{\subst\gA\alpha\gB}\sigma=\SemTyP{\gA}
  {\sigma,\alpha\mapsto\SemTyN{\gB}\sigma}$.\\
  Also note that $\SemTyP {\gA \cap \gB} \sigma = \SemTyP \gA \sigma \cap \SemTyP
  \gB \sigma$ and $\SemTyP {\forall \alpha \gA}\sigma = \bigcap_{Y\subseteq
    \ValDomL} \SemTyP \gA {\sigma, \alpha \mapsto Y}$.
\end{remark}

An orthogonality model is a sufficiently rich structure for Adequacy to hold:
\begin{toappendix}
  \begin{lemma}[Adequacy Lemma]\label{lem:Adequacy}\strut

    If $\DeriliS \gGamma M \gA$, then for all valuations $\sigma$ and for all
    mappings $\rho\in\SemTyP\gGamma\sigma$ we have
    $\SemTe M\rho\in\SemTyP \gA\sigma$.
  \end{lemma}
\end{toappendix}

\begin{toappendix}
  [
  \begin{proof}
    By induction on $\DeriliS \gGamma M \gA$, using the axioms (A1),\ldots,(A6) from Definition~\ref{prop}.
    See Appendix~\thisappendix.
  \end{proof}
  ]
  \begin{proof}
    By induction on the derivation of $\DeriliS \gGamma M \gA$, using the axioms (A1),\ldots,(A6) from Definition~\ref{prop}. Let $\sigma$ be a valuation.
    \begin{iteMize}{$\bullet$}
    \item $\infer{\Derili{\gGamma,x\col \gA} x \gA}{\strut}$\\
      Let $\rho\in\SemTyP {\gGamma,x\col\gA}\sigma$ and let $\v v\in\SemTyN \gA\sigma$.
      By definition, $\rho(x)\in\SemTyP \gA\sigma$ so $\orth {\rho(x)}{\v v}$, and by axiom (A1) we have $\orth{\SemTe x\rho}{\v v}$. Hence, ${\SemTe x\rho}\in \SemTyP \gA\sigma$.\bigskip

    \item $\infer{\Derili{\gGamma} {\lambda x.M} {\gA\arr \gB}}{\Derili{\gGamma,x\col \gA} M \gB}$\\
      Let $\rho\in\SemTyP \gGamma\sigma$ and let $\cons w{\v v}\in\SemTyN {\gA\arr\gB}\sigma=\cons{\SemTyP {\gA}\sigma}{\SemTyN {\gB}\sigma}$.
      As $w\in\SemTyP {\gA}\sigma$, we have $(\rho,x\mapsto w)\in\SemTyP {\gGamma,x\col\gA}\sigma$, so by induction hypothesis we have $\SemTe M{\rho,x\mapsto w}\in\SemTyP {\gB}\sigma$. From this we get $\orth {\SemTe{M}{\rho,x\mapsto w}}{\v v}$ and by axiom (A3) we have 
      $\orth {\SemTe{\lambda x.M}{\rho}}{\cons{w}{\v v}}$. Hence, ${\SemTe {\lambda x.M}\rho}\in \SemTyP {\gA\arr\gB}\sigma$.\bigskip

    \item $\infer{\Derili{\gGamma} {M\ N} \gB}{\Derili\gGamma M {\gA\arr \gB}\quad\Derili\gGamma N \gA}$\\
      Let $\rho\in\SemTyP \gGamma\sigma$ and let $\v v\in\SemTyN {\gB}\sigma$.
      By induction hypothesis we have $\SemTe N{\rho}\in\SemTyP {\gA}\sigma$ and $\SemTe M{\rho}\in\SemTyP {\gA\arr\gB}\sigma$. We thus get $\cons {\SemTe N{\rho}}{v}\in\cons {\SemTyP {\gA}\sigma}{\SemTyN {\gB}\sigma}=\SemTyN {\gA\arr\gB}\sigma$. So $\orth{\SemTe M{\rho}}{\cons{\SemTe N{\rho}}{v}}$ and by axiom (A2) we have
      $\orth{\SemTe {M\ N}{\rho}}{{v}}$. Hence $\SemTe {M\ N}{\rho}\in\SemTyP {\gB}\sigma$.\bigskip

    \item $\infer{\Derili \gGamma M {\gA\cap \gB}}{\Derili \gGamma M {\gA}\quad \Derili \gGamma M {\gB}}$\\
      Let $\rho\in\SemTyP \gGamma\sigma$ and let $\v v\in\SemTyN {\gA\cap\gB}\sigma=\SemTyN {\gA}\sigma\cup\SemTyN {\gB}\sigma$.
      By induction hypothesis we have $\SemTe M{\rho}\in\SemTyP {\gA}\sigma$ and $\SemTe M{\rho}\in\SemTyP {\gB}\sigma$ so in any case $\orth{\SemTe M{\rho}}{v}$. Hence $\SemTe M{\rho}\in\SemTyP {\gA\cap\gB}\sigma$.\bigskip

    \item $\infer{\Derili \gGamma M {\gA}}{\Derili \gGamma M {\gA\cap \gB}}$\\
      Let $\rho\in\SemTyP \gGamma\sigma$ and let $\v v\in\SemTyN {\gA}\sigma\subseteq\SemTyN {\gA\cap\gB}\sigma$.
      By induction hypothesis we have $\SemTe M{\rho}\in\SemTyP {\gA\cap\gB}\sigma$ so $\orth{\SemTe M{\rho}}{v}$. Hence $\SemTe M{\rho}\in\SemTyP {\gA}\sigma$.\bigskip

    \item $\infer{\Derili \gGamma M \gB}{\Derili \gGamma M {\gA\cap \gB}}$\\
      Let $\rho\in\SemTyP \gGamma\sigma$ and let $\v v\in\SemTyN {\gB}\sigma\subseteq\SemTyN {\gA\cap\gB}\sigma$.
      By induction hypothesis we have $\SemTe M{\rho}\in\SemTyP {\gA\cap\gB}\sigma$ so $\orth{\SemTe M{\rho}}{v}$. Hence $\SemTe M{\rho}\in\SemTyP {\gB}\sigma$.\bigskip

    \item {{$\infer[\alpha\notin\TV\gGamma]{\Derili \gGamma M {\forall \alpha \gA}}{\Derili \gGamma M {\gA}}$}}\\
      Let $\rho\in\SemTyP \gGamma\sigma$ and let $\v v\in\SemTyN {\forall\alpha\gA}\sigma=\bigcup_{Y\subseteq\ValDomL}\SemTyN {\gA}{\sigma,\alpha\mapsto Y}$.
      By induction hypothesis we have $\SemTe M{\rho}\in\SemTyP {\gA}{\sigma,\alpha\mapsto Y'}$ for all $Y'\subseteq\ValDomL$, so in any case $\orth{\SemTe M{\rho}}{v}$. Hence $\SemTe M{\rho}\in\SemTyP {\forall\alpha\gA}\sigma$.\bigskip

    \item $\infer{\Derili \gGamma M {\subst \gA \alpha \gB}}{\Derili \gGamma M {\forall \alpha \gA}}$\\
      Let $\rho\in\SemTyP \gGamma\sigma$ and let $\v v\in\SemTyN {\subst \gA\alpha\gB}\sigma=\SemTyN {\gA}{\sigma,\alpha\mapsto\SemTyN {\gB}{\sigma}}\subseteq\SemTyN {\forall\alpha\gA}\sigma$.
      By induction hypothesis we have $\SemTe M{\rho}\in\SemTyP {\forall\alpha\gA}\sigma$ so $\orth{\SemTe M{\rho}}{v}$. Hence $\SemTe M{\rho}\in\SemTyP {\subst \gA\alpha\gB}\sigma$.\bigskip

    \item     $\infer{\Derili \gGamma {\Subst M x N} \gB}{\Derili \gGamma N \gA \quad \Derili {\gGamma, x \col \gA} M \gB}$\\
      Let $\rho\in\SemTyP \gGamma\sigma$ and let $\v v\in\SemTyN \gB\sigma$. By induction hypothesis we have $\SemTe N{\rho}\in\SemTyP {\gA}\sigma$; therefore $\SemTe N{\rho}$ is a value and $(\rho,x\mapsto \SemTe N{\rho})\in\SemTyP {\gGamma, x\col \gA}\sigma$. By induction hypothesis again we have $\SemTe M{\rho,x\mapsto \SemTe N{\rho}}\in\SemTyP {\gB}\sigma$. So $\orth{\SemTe M{\rho,x\mapsto \SemTe N{\rho}}}{\v v}$ and by axiom (A4) we have  $\orth{\SemTe {\Subst M x N}{\rho}}{\v v}$.\bigskip

    \item $\infer{\Derili {\gGamma, x \col \gA} {\weakening x M} \gB} {\Derili \gGamma M \gB \quad x \notin dom(\gGamma)}$\\
      Let $(\rho,x\mapsto u)\in\SemTyP {\gGamma, x \col \gA}\sigma$ and let $\v v\in\SemTyN \gB\sigma$. We have $\rho\in\SemTyP {\gGamma}\sigma$ and by induction hypothesis we have $\SemTe M{\rho}\in\SemTyP {\gB}\sigma$. So $\orth{\SemTe M{\rho}}{\v v}$ and by axiom (A5) we have $\orth{\SemTe {\weakening x M}{\rho,x\mapsto u}}{\v v}$.\bigskip

    \item $\infer{\Derili {\gGamma, x \col \gA} {\contraction x y z M} \gB}
      {\Derili {\gGamma, y \col \gA, z \col \gA} M \gB}$\\
      Let $(\rho,x\mapsto u)\in\SemTyP {\gGamma, x \col \gA}\sigma$ and let $\v v\in\SemTyN \gB\sigma$. We have $(\rho,y\mapsto u,z\mapsto u)\in\SemTyP {\gGamma}\sigma$ and by induction hypothesis we have $\SemTe M{\rho,y\mapsto u,z\mapsto u}\in\SemTyP {\gB}\sigma$. So $\orth{\SemTe M{\rho,y\mapsto u,z\mapsto u}}{\v v}$ and by axiom (A6) we have $\orth{\SemTe {\weakening x M}{\rho,x\mapsto u}}{\v v}$.\qedhere
    \end{iteMize}
  \end{proof}
\end{toappendix}

\subsubsection{The special case of applicative structures}
\label{sec:appstructures}

In the next section we present instances of orthogonality models.  They will have in common that $\ValDomE$ is an applicative structure, as we have seen with I-filters. This motivates the following notion:
\begin{definition}[Applicative orthogonality model]\strut\\
  An {\em applicative orthogonality model} is a 4-tuple
  $(\ValDomE,\ValDom,\ap,\SemTe {\_ }{\_} )$ where:
  \begin{iteMize}{$\bullet$}
  \item $\ValDomE$ is a set, $\ValDom$ is a subset of $\ValDomE$, $\ap$ is a
    (total) function from $\ValDomE\times\ValDomE$ to $\ValDomE$, and
    $\SemTe {\_ }{\_}$ is a function (parametrised by an environment) from
    $\lambda$-terms to the support.
  \item $(\ValDomE,\ValDom,\orth{}{},\SemTe {\_ }{\_} )$ is an orthogonality
    model,\\
    where the relation $\orth u {\vec v}$ is defined as $(u\ap \vec{v})\in \ValDom$
    \\
    (writing $u\ap\vec v$ for $(\ldots(u\ap v_1)\ap\ldots \ap v_n)$ if
    $\vec v=\cons{v_1}{\ldots\cons{v_n}\el}$).
  \end{iteMize}
\end{definition}

\begin{remark} 
Axioms (A1) and (A2) are ensured provided that $\SemTe {M\ N} \rho=\SemTe {M}\rho\ap \SemTe {N}\rho$ and $\SemTe {x}\rho=\rho(x)$.  These conditions can hold by definition (as in term models, \cf the next Section), or can be proved (as in Theorem~\ref{th:SemTe}, which also proves axioms (A3)-(A6)).
\end{remark}

\subsubsection{Instances of orthogonality models}
\label{sec:instances}

We now give instances of (applicative) orthogonality models with well-chosen sets of values, applications, and interpretations of terms, with the aim of deriving the strong normalisation of a term $M$ of type $\gA$ in $\SysS$ from the property $\SemTe M{}\in\SemTyP {\gA}{}$.

The first two instances are term models: terms are interpreted as pure $\lambda$-terms (see Definition~\ref{def:termint}), so the support of those term models is the set of all pure $\l$-terms seen as an applicative structure (using term application: ${M_1}\aptermmodel {M_2}\eqdef M_1\ M_2$).

\begin{definition}[Interpretation of terms in a term model]\strut
  \label{def:termint}
  \[\begin{array}{lll}
    \SemTe[\termmodel] x \rho&\eqdef \rho(x)\\
    \SemTe[\termmodel] {M_1\ M_2} \rho&\eqdef \SemTe[\termmodel] {M_1} \rho\
    \SemTe[\termmodel] {M_2} \rho\\
    \SemTe[\termmodel] {\lambda x.M} \rho&\eqdef \lambda x.\SemTe[\termmodel] {M} {\rho,x\mapsto x}\\
    \SemTe[\termmodel] {\Subst M x N} \rho&\eqdef {\SemTe[\termmodel] {M} {\rho,x\mapsto \SemTe[\termmodel] {N}\rho}}\\
    \SemTe[\termmodel] {\weakening x M} \rho&\eqdef \SemTe[\termmodel] {M} {\rho}\\
    \SemTe[\termmodel] {\contraction x y z M} \rho&\eqdef \SemTe[\termmodel] {M} {\rho,y\mapsto \rho(x),z\mapsto \rho(x)}
  \end{array}
  \]
\end{definition}

\begin{remark}
  $\SemTe[\termmodel]{\subst M x N}\rho=\SemTe[\termmodel] M{\rho,x\mapsto
    \SemTe[\termmodel]{N}\rho}$ 
\end{remark}

In the first instance, values are those pure $\lambda$-terms that are strongly normalising (for $\beta$).  If we concentrate on the interpretation (as themselves) of the pure $\l$-terms that are typed in $\SysS$, we have an orthogonality model that rephrases standard proofs of strong normalisation by orthogonality or reducibility candidates~\cite{Parigot97,LM:APAL07}.

In the second instance, values are those pure $\lambda$-terms that can be typed with intersection types, for instance in system~\InterLambdaSub.

\begin{theorem}
  The structures
  \begin{iteMize}{$\bullet$}
  \item $\MSN\eqdef(\LambdaTerms,\SNLambda{},\aptermmodel,\SemTe[\termmodel]\_\_)$
  \item $\Mcap\eqdef(\LambdaTerms,\TypableTerms,\aptermmodel,\SemTe[\termmodel]\_\_)$
    \hfill (where $\TypableTerms$ is the set of pure $\l$-terms typable in
    \InterLambdaSub)
  \end{iteMize}
  are applicative orthogonality models.
\end{theorem}

Indeed, the applicative structures $\MSN$ and $\Mcap$ already satisfy axioms (A1), (A2), and (A4) to (A6), because of Definition~\ref{def:termint}.
Axiom (A3) holds in $\MSN$ and $\Mcap$ because of their respective expansion properties:
\begin{lemma}[Expansion]\label{lem:expansion}\strut
  \begin{enumerate}[\em(1)]
  \item
    If $\subst M x P\ \vec N \in \SNLambda{}$ and $P\in \SNLambda{}$, then $(\lambda x.M)\ P\ \vec N \in \SNLambda{}$.
  \item
    If $\subst M x P\ \vec N \in \TypableTerms$ and $P\in \TypableTerms$, then $(\lambda x.M)\ P\ \vec N \in \TypableTerms$.
  \end{enumerate}
\end{lemma}

Admittedly, once \InterLambdaSub\ has been proved to characterise $\SNLambda{}$ (Theorems~\ref{th:soundness} and~\ref{th:lambda-compl}), the two points are identical and so are the two models $\MSN$ and $\Mcap$.  But involving the bridge of this characterisation goes much beyond what is needed for either point: point 1 is a known fact of the literature; point 2 is a simple instance of Subject Expansion (Theorem~\ref{th:SubjectExpansion} in the next section) not requiring Subject Reduction (Theorem~\ref{th:SubjectReduction}) while both are involved at some point in the more advanced property $\SNLambda{}=\TypableTerms$.  In brief, as we are interested in comparing proof {\em techniques} for the strong normalisation of System~$\SysS$, the question of which properties are used and in which order matters.

Now using the Adequacy Lemma (Lemma~\ref{lem:Adequacy}), we finally get: 
\begin{corollary}
  If $\DeriliS \gGamma M \gA$, then:
  \begin{enumerate}[$\MSN$]
  \item[$\MSN$]
    For all valuations $\sigma$ and all mappings $\rho\in\SemTyP\gGamma\sigma$ we
    have $\SemTe[\termmodel] M\rho\in \SNLambda{}$.
  \item[$\Mcap$]
    For all valuations $\sigma$ and all mappings $\rho\in\SemTyP\gGamma\sigma$\\
    there exist $\Gamma$ and $A$ such that
    $\Derilis \Gamma {\SemTe[\termmodel] M\rho} A$, and therefore $\SemTe[\termmodel] M\rho\in\SNLambda$.
  \end{enumerate}
\end{corollary}
\noindent For $\Mcap$ we conclude of course by using Theorem~\ref{th:soundness}.

Now notice that, if $M$ is a pure $\lambda$-term, this entails $M\in \SNLambda{}$ by choosing, in both models, $\sigma$ to map every type variable to the empty set, and $\rho$ to map every term variable to itself. Indeed we have:
\begin{remark}
    In both structures $\MSN$ and $\Mcap$ we can check that:\\
    For all lists $\vec N$ of values, and any term variable $x$, $\orth x {\vec N}$.\\
    Hence, for all valuations $\sigma$ and all types $\gA$, $x\in\SemTyP \gA
    \sigma$.
\end{remark}

However if $M$ is not a pure $\lambda$-term, it is not obvious to derive an interesting normalisation result for $M$, given that the explicit substitutions / weakenings / contractions in $M$ have disappeared in $\SemTe[\termmodel] M\rho$ (and in the case of $\MSN$, relating $\SNLambda$ to $\SNLambdaS$ or $\SNLambdaLxr$ is another task to do).

An idea would be to tweak the interpretation of terms so that every term is interpreted as itself, even if it has explicit substitutions / weakenings / contractions:
  \[\begin{array}{lll}
    \SemTe[\termmodel] {\Subst M x N} \rho&\eqdef \Subst {\SemTe[\termmodel] {M} {\rho,x\mapsto x}} x {\SemTe[\termmodel] {N}\rho}\\
    \SemTe[\termmodel] {\weakening x M} \rho&\eqdef \weakening x {\SemTe[\termmodel] {M} {\rho}}\\
    \SemTe[\termmodel] {\contraction x y z M} \rho&\eqdef \contraction x y z {\SemTe[\termmodel] {M} {\rho,y\mapsto y,z\mapsto z}}
  \end{array}
  \]
But proving axioms (A1) to (A6) then becomes much more difficult.
This is however the direction taken by~\cite{Kesner09lmcs} for the explicit substitution calculus $\lambda_{\textsf{ex}}$, where the methodological cornerstone is a property called {\bf IE}, which is nothing else but axiom (A4) in $\MSN$. For $\Mcap$, it might be possible to prove the axioms by inspecting typing derivations and/or using Subject Expansion (Theorem~\ref{th:SubjectExpansion} in the next section).

A quicker way is to depart from term models and turn the filter model \Mfiltersi\ from Section~\ref{sec:ifiltermodel} into an orthogonality model: a term is interpreted as the filter of the intersection types that it can be assigned (\eg in \InterLambdaSub, see Definition~\ref{def:intterms}), and orthogonality is defined in terms of filters being non-empty.

Strong Normalisation will then follow, in a very uniform way for the three calculi $\lambda$, $\LambdaS$, and $\LambdaLxr$, from the fact that terms typable with intersection types are themselves strongly normalising (Theorems~\ref{th:soundness},~\ref{th:LSsoundness},~\ref{th:Llxrsoundness} for \InterLambdaSub).
\begin{theorem}
  The structure $\Mfilters\eqdef(\ValDomE,\ValDom,\ap,\SemTe\_\_)$ (with the four components as defined in Section~\ref{sec:filters}) is an applicative orthogonality model.
\end{theorem}
\begin{proof}
  Indeed, $\Mfilters$ satisfies axioms (A1) to (A6) as immediate consequences of Theorem~\ref{th:SemTe}.
\end{proof}

\begin{remark}
    For $\Mfilters$ we now have:
    For all list of values $\vec v$, $\orth \top {\vec v}$.\\ Hence, for all
    valuations $\sigma$ and all types $\gA$, $\top\in\SemTyP \gA \sigma$.
\end{remark}

Now using the Adequacy Lemma (Lemma~\ref{lem:Adequacy}), we finally get: 
\begin{corollary}
  If $\DeriliS \gGamma M \gA$, then:\\
    For all valuations $\sigma$ and all mappings $\rho\in\SemTyP\gGamma\sigma$ we
    have $\SemTe M\rho\neq \bottom$.\\
    Hence, there exist $\Gamma$ and $A$ such that $\Derilis \Gamma M A$.\\
    Finally, $M\in \SNLambda{}$ (\resp $M\in \SNLambdaS{}$, $M\in \SNLambdaLxr{}$, according to the calculus considered).
\end{corollary}

\begin{proof}\strut
    The first statement holds because $\bottom\notin\ValDom$ and
    $\SemTyP \gA\sigma\subseteq\ValDom$.
    To prove the second, we need to show that there exist such a $\sigma$ and such
    a $\rho$; take $\sigma$ to map every type variable to the empty set and take
    $\rho$ to map every term variable to $\top$.
    The final result comes from Theorem~\ref{th:soundness} (\resp Theorem~\ref{th:LSsoundness}, Theorem~\ref{th:Llxrsoundness}).
\end{proof}

\section{Completeness}

\label{sec:completeness}

In Section~\ref{sec:soundness} we have shown that for the three calculi, if a term is typable with intersection types, then it is strongly normalising.  We have briefly mentioned that the converse is true.  In this section we give a proof for the three calculi.  Moreover, the typing trees obtained by these completeness theorems satisfy some interesting properties that will be used in the next sections (for the complexity results): they are \emph{optimal} and \emph{principal}.

The proof of completeness for $\LambdaS$ is simpler than the one for the pure $\lambda$-calculus. This is why, in this section, that we will treat the $\LambdaS$ calculus first.

\subsection{Two properties of typing trees: Optimality and Principality}

In the next sub-section we will notice that the typing trees produced by the
proof of completeness all satisfy a particular properties. 
In this section we define these properties.
The first of these is {\em optimality}.

This property involves the following notions:
\begin{definition}[Subsumption and forgotten types]\strut
  \begin{iteMize}{$\bullet$}
  \item If $\pi$ is a typing tree, we say that $\pi$ {\em does not use subsumption} if
    it features an occurrence of the abstraction rule where the condition
    $A \subseteq U$ is either $A \approx U$ or $A \subseteq \tempty$.
  \item We say that a type $A$ is {\em forgotten} in an instance of rule
$(\mbox{Abs})$ or rule $(\mbox{Subst})$ if in the side-condition of the rule we
have $U=\omega$.
  \item If a typing tree $\pi$ uses no subsumption, we collect the list of its
forgotten types, written $\forg\pi$, by a standard prefix and depth-first search
of the typing tree $\pi$.
  \end{iteMize}

\end{definition}

\noindent The optimal property also involves refining the grammar of types: 
\begin{definition}[Refined intersection types]
  $A^+$, $A^-$, $A^{--}$ and $U^{--}$ are defined by the following grammar:
  $$\begin{array}{ll}
    A^+, B^+ &\recdef \tau \sep A^{--} \rightarrow B^+\\
    A^{--}, B^{--} &\recdef A^- \sep A^{--} \cap B^{--}\\
    A^-, B^- &\recdef \tau \sep A^+ \rightarrow B^- \\
		U^{--}, V^{--} & \recdef A^{--} \sep \tempty
  \end{array}$$
We say that $\Gamma$ is of the form $\Gamma^{--}$ if for all $x$, $\Gamma(x)$ is
of the form $U^{--}$.
  The {\em degree} of a type of the form $A^+$ is the number of arrows in
negative positions:
  $$\begin{array}{|ll|}
    \hline
    \degr[+]{\tau}&\eqdef 0\Astrut \\
    \degr[+]{A^{--} \rightarrow B^+}&\eqdef \degr[-]{A^{--}}+\degr[+]{B^+}+1
    \Bstrut\\
    \hline
    \degr[-]{A^{--} \cap B^{--}}&\eqdef \degr[-]{A^{--}}+\degr[-]{B^{--}}
    \Cstrut\\
    \hline
    \degr[-]{\tau}&\eqdef 0\Astrut \\
    \degr[-]{A^+ \rightarrow B^-}&\eqdef \degr[+]{A^+}+\degr[-]{B^-}\Bstrut \\
    \hline
  \end{array}
  $$
\end{definition}

\begin{remark}
  If $A^+\approx B^+$, then $\degr[+]A=\degr[+]B$, and if $A^{--}\approx B^{--}$, then $\degr[-]{A^{--}}=\degr[-]{B^{--}}$.
\end{remark}

We can finally define the optimal property:
\begin{definition}[Optimal typing]
  A typing tree $\pi$ concluding $\Deri \Gamma M A$ is optimal if 
  \begin{iteMize}{$\bullet$}
  \item There is no subsumption in $\pi$
  \item $A$ is of the form $A^+$
  \item For every $(x : B) \in \Gamma$, $B$ is of the form $B^{--}$
  \item For every forgotten type $B$ in $\pi$, $B$ is of the form $B^+$.
  \end{iteMize}
  We write $\Deriliopt \Gamma M {A^+}$ if there exists such $\pi$.

  The {\em degree} of such a typing tree is defined as
  $$\degr\pi=\degr[+]{A^+}+\Sigma_{x~\col~ {B^{--}}\in\Gamma}\degr[-]{B^{--}}+
\Sigma_{C^+\in\forg\pi}\degr[+]{C^+}$$
\end{definition}

In this definition, $A^+$ is an output type, $A^-$ is a basic input type (\ie
for a variable to be used once), and $A^{--}$ is the type of a variable that
can be used several times. The intuition behind this asymmetric grammar can be
found in linear logic:

\begin{remark}
	Intersection in a typing tree means duplication of resource.  So intersections can be compared to exponentials in linear logic~\cite{girard-ll}.  Having an optimal typing tree means that duplications are not needed in certain parts of the optimal typing tree.  In the same way, in linear logic, we do not need to have exponentials everywhere: A simple type $T$ can be translated as a type $T^*$ of linear logic as follows:
  $$\begin{array}{|ll|}
    \hline
    \tau^* &\eqdef \tau\Astrut\\
    (T \rightarrow S)^* &\eqdef ! T^* \multimap S^*\Bstrut\\
    \hline
  \end{array}
  $$
We can find a more refined translation; it can also be translated as $T^+$ and $T^-$ as follow :
  $$\begin{array}{|ll@{\quad}|@{\quad}ll|}
    \hline
    \tau^+ &\eqdef \tau&\tau^- &\eqdef \tau\Astrut\\
    (T \rightarrow S)^+ &\eqdef ! T^- \multimap S^+&
    (T \rightarrow S)^- &\eqdef T^+ \multimap S^-\Bstrut\\
    \hline
  \end{array}$$
  And we have in linear logic : $T^- \vdash T^*$ and $T^* \vdash T^+$. 
So the translation $T^+$ is sound and uses less exponentials that the usual
and naive translation. In some way, it is more ``optimal''.
The main drawback is that we cannot compose proofs of optimal translations
easily.

\end{remark}

We now introduce the second of these properties: the notion of principal typing.


\begin{definition}[Principal typing]\strut

A typing tree $\pi$ of $M$ is \emph{principal}\footnote{
  In the literature, \emph{principality} is often used (see \eg\cite{Wells02theessence}) for a typing judgement that characterises, for a given term, all of its derivable typing judgements:
they are all obtained from the principal one by instantiation of its type variables.
This notion was our guiding intuition (and we therefore kept the terminology) and does relate to a notion of minimality in the size of types (a type is smaller in size than its instances). But for our purpose, we needed a notion that also says something about typing trees rather than the typing judgements that they derive, so optimality of typing trees led to a notion of principality on trees rather than judgements. Saying that all (optimal) typing trees can be obtained by instantiation of the principal one would require quotienting the trees to identify irrelevant differences (integrating AC-equivalence of types and contexts), so it was simpler to express a minimality condition on degrees, both to define the notion and to use it in Section ~\ref{sec:complexity} for Complexity results.
}
 if it is optimal and of minimal degree: For every optimal typing tree $\pi'$ of $M$, $\degr \pi \leq \degr {\pi'}$.

\end{definition}

\subsection{$\LambdaS$}

In order to prove the completeness of the typing system with respect to $\SNLambdaS$, we first show that terms in normal form (for some adequate notion of normal form) can be typed, and then we prove Subject Expansion for a notion of reduction that can reduce any term in $\SNLambdaS$ to a normal form (which we know to be typed).
In $\LambdaS$, Subject Expansion is true only for $\Rew{B, S}$ (not for $\Rew{W}$).  We will prove that it is enough for completeness. The main reason is that $\Rew{W}$ can be postponed \wrt $\Rew{B,S}$:

\begin{lemma}[Erasure postponement]\strut
  \label{lem:erasurepostponed}

\begin{iteMize}{$\bullet$}

\item If $M \Rew{W} \Rew{B} M'$, then $M \Rew{B} \Rew{W} M'$.

\item If $M \Rew{W} \Rew{S} M'$, then $M \Rew[+]{S} \Rew[+]{W} M'$.

\item If $M \Rewn{S,W} M'$, then $M \Rewn{S} \Rewn{W} M'$

\end{iteMize}


\end{lemma}
\begin{proof}

The first two points are proved by inspection of the rules:
a substitution never blocks computation.

The third point: Let $L$ a $S,W$ reduction sequence from $M$ to $M'$.

If $L$ is not a of the form $\Rewn{S} \Rewn{W}$, then there exists
$\Rew{W} \Rew{S}$ in $L$ and by using the second point we can replace it
by $\Rew[+]{S} \Rew[+]{W}$ to obtain a reduction sequence $L'$.
Therefore we have a non-deterministic rewriting of $L$.

This rewriting increases or does not change the size of $L$.
According to Lemma~\ref{lem:SNForSW}, $M$ is strongly normalising for $S,W$.
Therefore, after a certain number of steps, the size of $L$ does not change.
So, after a certain number of steps, the rewriting is just replacing
$\Rew{W} \Rew{S}$ by $\Rew{S} \Rew{W}$ and this terminates.
Hence, this rewriting terminates.

By taking a normal form of this rewriting we have $M \Rewn{S} \Rewn{W} M'$.
\end{proof}

Therefore, the normal forms for $\Rew{B, S}$ are ``normal enough'' to be easily typed:
\begin{lemma}[Typability of $B,S$-normal term]\strut

  If $M$ cannot be reduced by $\Rew{B, S}$, then there exist $\Gamma$ and $A$ such
  that $\Deriliopt \Gamma M A$.
\end{lemma}

\begin{proof}
  By induction on $M$.

  We use the fact that if $M$ cannot be reduced by $\Rew{B,S}$, then
  $M$ is of one of the following form:
  \begin{iteMize}{$\bullet$}
  \item $\l x . M_1$
  \item $\Subst {\Subst x {y_1} {M_1} \cdots} {y_n} {M_n} N_1 \cdots N_n$
  \end{iteMize}
  Each of them can easily be typed by a principal typing tree using the induction
  hypothesis.
\end{proof}

\begin{remark}
  The algorithm given by the previous proof gives us a principal typing tree.
\end{remark}

\begin{theorem}[Subject Expansion]
  \label{th:SELambdaS}\strut

  If $M \Rew{B, S} M'$ and $\Derilis {\Gamma'} {M'} A$, then there exist $\Gamma\approx \Gamma'$ such that $\Derilis \Gamma M A$.

Moreover, the optimality property, the degree, and the principality property are all preserved.
\end{theorem}

\begin{proof}
  First by induction on $\Rew{B, S}$ and $\approx$, then by induction on $A$.

  We adapt the proof of Subject Reduction.
  The optimality property, the degree, and the principality property are preserved in both directions (Subject Expansion and Subject Reduction): indeed, since we are considering $\Rew{B,S}$ and not $\Rew{W}$, the interface (typing context, type of the term and forgotten types) is not changed and we do not add any subsumption.
\end{proof}

\begin{theorem}[Completeness]  If $M\in\SNLambdaS$, then there exist $\Gamma$ and $A$ such that $\Deriliopt \Gamma M A$.
\end{theorem}

\begin{proof}
  By induction on the longest reduction sequence of $M$.
  If $M$ can be reduced by $\Rew{B,S}$ we can use the induction hypothesis.
  Otherwise, $M$ is typable.
\end{proof}

\begin{remark}
  The algorithm given by the previous proof gives us a principal typing tree.
\end{remark}

\begin{corollary}
  If $M \Rew{B, S} M'$ and $M'\in\SNLambdaS$, then $M\in\SNLambdaS$.
\end{corollary}

\subsection{Pure $\l$-calculus}

\label{lambda-compl}

As in the case of $\LambdaS$, proving completeness of the typing system with respect to $\SNLambda$ relies on the typability of (some notion of) normal forms and on some property of Subject Expansion.

For the $\lambda$-calculus, the normal forms that we consider are simply the $\beta$-normal forms. Typing them with the optimal typing trees of System~\InterLambdaSub\ is therefore very reminiscent of~\cite{DHM05} that applies a similar technique to type $\beta$-normal forms with the principal types of a system with idempotent intersections.

\begin{definition}[Accumulators]

  The fact that a $\l$-term $M$ is a $\l$-free head-normal form with head-variable $x$ (\ie $M$ is of the form $x M_1 \ldots M_n$) is abbreviated as $\accu M x$. Equivalently, the judgement $\accu M x$ can be defined with the following rules:
  \[
  \begin{array}{c}
    \infer{\accu x x}{\strut} \quad
    \infer{\accu {M N} x}{\accu M x}\\
  \end{array}
  \]

\end{definition}

\begin{remark}[Shape of a normal term]\strut
  If $M$ is a $\beta$-normal form, then either $M$ is of the form $\l x . N$ (for some normal form $N$) or there
  exists $x$ such that $\accu M x$ (induction on $M$ -see \eg \cite{Bohm68}).
\end{remark}


\begin{lemma}[Typability of accumulators]\strut\\
If $\accu M x$ and $\pi$ is a derivation of $\Derilis {\Gamma,x\col U^{--}} M F$, then
  \begin{enumerate}[\em(1)]
  \item $F$ is of the form $F^{-}$;
  \item for all $G^-$, there exists $V^{--}$ and a derivation $\pi'$ of $\Derilis {\Gamma, x \col V^{--}} M {G^-}$;\\
    moreover, $\forg{\pi'}=\forg \pi$ and if $\pi$ does not use subsumption, neither does $\pi'$.
  \end{enumerate}
\end{lemma}

\begin{proof}

  \begin{enumerate}[(1)]

  \item By induction on $\accu M x$.

  \item
    By induction on $\accu M x$.
    \begin{iteMize}{$\bullet$}

    \item For \infer{\accu x x}{}:
      Then $\Gamma^{--} = ()$ and $\Derilis {x \col G^-} x {G^{-}}$.

    \item For \infer{\accu {M N} x}{\accu M x}:
      Then, there exist $\Gamma_1^{--}$, $\Gamma_2^{--}$, $U_1^{--}$, $U_2^{--}$ and
      $A$ such that
      $\Gamma^{--} = \Gamma_1^{--} \cap \Gamma_2^{--}$,
      $U^{--} = U_1^{--} \cap U_2^{--}$,
      $\Derilis {\Gamma_1^{--}, x \col U_1^{--}} M {A \rightarrow F}$ and
      $\Derilis {\Gamma_2^{--}, x \col U_2^{--}} N A$.
      By the first point, $A \rightarrow F$ is of the form $B^-$.
      Therefore, $A$ is of the form $A^+$.
      Hence, $A^+ \rightarrow G^-$ is of the form $C^-$.
      By induction hypothesis, there exist $V_1^{--}$ such that
      $\Derilis {\Gamma_1^{--}, x \col V_1^{--}} M {A^+ \rightarrow G^-}$.
      Therefore, $\Derilis{\Gamma, x \col V_1^{--} \cap U_2^{--}} {M N} {G^-}$.\qedhere

    \end{iteMize}
  \end{enumerate}
\end{proof}

\begin{lemma}[Typability of a normal term]\strut

  If $M$ is a normal form, then there exists $\Gamma$ and $F$ such that
  $\Deriliopt \Gamma M F$.

\end{lemma}

\begin{proof}
  By induction on $M$. Since $M$ is a normal form, we are in one of the following cases:
  \begin{iteMize}{$\bullet$}
  \item $M$ is of the form $\l x . N$ ; by the induction hypothesis we can type $N$ so we can type $\l x . N$;
  \item There exists $x$ such as $\accu M x$ and
    \begin{iteMize}{$-$}
    \item either $M = x$, which can trivially be typed;
    \item or $M = M_1 M_2$ with $\accu {M_1} x$, and by the induction hypothesis we can type $M_1$ and $M_2$;
      therefore we can give any type to $M_1$ so we can type $M_1 M_2$.\qedhere
    \end{iteMize}
  \end{iteMize}
\end{proof}

\begin{lemma}[Anti-substitution lemma]
  If $\Derilis \Gamma {\subst M x N} A$, then there exist $\Gamma_1$, $\Gamma_2$ and $U$ such that $\Derilis {\Gamma_1} N U$, $\Derilis {\Gamma_2} M A$ and $\Gamma_1 \cap \Gamma_2 \approx \Gamma$.
\end{lemma}

\begin{proof}
  First by induction on $M$ then by induction on $A$.  We adapt the proof of Lemma~\ref{lem:typsubSR}.  Notice that if $x \notin fv(M)$ we take $U = \tempty$. Therefore, $N$ might not be typable by an $A$-type.
\end{proof}

As we have seen for $\LambdaS$, proving completeness relies on the Subject Expansion property for a notion of reduction that can reduce any term in $\SNLambda$ to a normal form (which we know to be typed).

In the pure $\lambda$-calculus, not all $\beta$-reductions satisfy Subject Expansion. For example, in the following reduction:
\[ (\l z . (\l y . a)(z z))(\l y . y y) \Rew{\beta} (\l z . a) (\l y . y y)\]
the second term is typable, but not the first one (because it is not strongly normalising).

As in $\LambdaS$, it is erasure that breaks the Subject Expansion ($(\l x . M) N \Rew{} M$ with $x \notin fv(M)$).

The problem here is that we cannot just study Subject Expansion for $\beta$-reductions that do not erase, because forbidding erasure can block a reduction sequence (for example, $(\l x . (\l y . y)) a b$).

So we have to define a restricted version of $\beta$-reduction that satisfies Subject Expansion, but that is still general enough to reach $\beta$-normal forms (which can be easily typed).

If $M$ and $N$ are $\l$-terms and $E$ a finite set of variables then we define
the judgements $M \per_E N$  and $M \Per_E N$ with the rules of Fig.~\ref{fig:restrictbeta}.

\begin{figure}[!h]
  $$\begin{array}{|c|}
    \hline\\
    \infer{(\l x . M) N \per_{\o} \subst M x N}{x \in fv(M)}\quad
    \infer{(\l x . M) N \per_{fv(N)} M}{x \notin fv(M) \quad N \mbox{ is a $\beta$-normal form}}
			\quad
    \infer{(\l x . M) N \per_E (\l x . M) N'}{x \notin fv(M) \quad N \Per_E N'}\\\\
    \infer{M N \per_E M'N}{M \per_E M'}\quad
    \infer{M N \per_E MN'}{N \per_E N'}\quad
    \infer{\l x . M \per_E \l x . M'}{M \per_E M' \quad x \notin E}\\\\
    \infer{M \Per_E M'}{M \per_E M'}\quad
    \infer{\l x . M \Per_{E - \{x\}} \l x . M'}{M \Per_E M'}\quad
    \infer{M N \per_{E \cup \{x\}} M N'}{\accu M x \quad N \Per_{E} N'}\\\\
    \hline
  \end{array}
  $$
  \caption{Restricted $\beta$-reduction}
  \label{fig:restrictbeta}
\end{figure}

These may not seem natural on a syntactic point of view.
However, they are quite intuitive if you consider that they satisfy the
following lemma:

\newpage
\begin{theorem}[Subject Expansion]\strut

  \label{th:SubjectExpansion}
  Assume $E = \{x_1, \ldots, x_n\}$ and
  $\Derilis {\Gamma, x_1 \col  U_1, \ldots , x_n \col  U_n} {M'} A$.

  \begin{iteMize}{$\bullet$}

  \item If  $M \Per_E M'$ then there exist $B$, $\Gamma'$, $V_1$, \ldots, $V_n$ such
    that $\Derilis {\Gamma', x_1 \col  V_1, \ldots , x_n \col  V_n} M B$ and
    $\Gamma \approx \Gamma'$.

  \item If $M \per_E M'$ then there exists $\Gamma'$, $V_1$, \ldots, $V_n$, such that
    $\Derilis {\Gamma', x_1 \col  V_1, \ldots, x_n \col  V_n} M A$ and
    $\Gamma \approx \Gamma'$.

  \end{iteMize}

  Moreover, if the typing of $M'$ is optimal, then the typing of $M'$ can be required to be optimal.
\end{theorem}

\begin{proof}
  First by induction on $M \per_E M'$ and $M \Per_E M'$ then by induction on $A$.  We adapt the proof of Subject Reduction (Theorem~\ref{th:SubjectReduction}).
\end{proof}

\begin{lemma}[Safe execution of a term]  If $M$ can be reduced by $\Rew{\beta}$ then there exist $M'$ and $E$ such that $M \Per_E M'$ and if $M$ is not of the form $\l x . M$ then $M \per_E M'$.
\end{lemma}

\begin{proof}
  By induction on $M$.
  \begin{iteMize}{$\bullet$}
  \item $M$ cannot be a variable.
  \item If $M$ is of the form $\l x . N$.
    Then $N$ reduced by $\Rew{\beta}$.
    By induction hypothesis there exist $N'$ and $E$ such that $N \Per_E N'$.
    Therefore $\l x . N \Per_{E - \{ x \}} \l x . N'$.
  \item If $M$ is of the form $M_1 M_2$.
    We are in one of the following cases:
    \begin{iteMize}{$-$}
    \item $M_1$ is of the form $\l x . M_3$.
      Then we have $(\l x . M_3) M_2 \per_E \subst {M_3} x {M_2}$ with
      $E = \emptyset$ or $E = fv(M_2)$.
    \item $M_1$ is not of the form $\l x . M_3$ and $M_1$ reduced by $\Rew{\beta}$.
      By induction hypothesis, there exist $M_1'$ and $E$ such that
      $M_1 \per_E M_1'$.
      Therefore $M_1 M_2 \per_E M_1' M_2$.
    \item $M_1$ is not of the form $\l x . M_3$ and $M_1$ is a $\beta$-normal form.
      Therefore there exists $x$ such that $\accu {M_1} x$ and $M_2$ reduced by $\Rew{\beta}$.
      By induction hypothesis, there exist $M_2'$ and $E$ such that
      $M_2 \Per_E M_2'$.
      Hence $M_1 M_2 \per_{E \cup \{ x \}} M_1 M_2'$.\qedhere
    \end{iteMize}
  \end{iteMize}
\end{proof}

\begin{theorem}[Completeness]
  \label{th:lambda-compl}
  If $M\in\SNLambda$ then there exists $\Gamma$ and $A$ such that
  $\Deriliopt \Gamma M A$.
\end{theorem}

\begin{proof}
  By induction on the longest reduction sequence of $M$.
\end{proof}

We can notice that to prove the completeness of the pure $\l$-calculus we only
need a fragment of $\per$ and $\Per$ but by dealing with all $\per$ and $\Per$
we have the following result without extra difficulties:

\begin{corollary}
  If $M \Per_E M'$ and $M'\in\SNLambda$ then $M\in\SNLambda$.
\end{corollary}

This result is purely syntactic. However, it is very hard to prove without
intersection types (if we consider all $\per$ and $\Per$ and not just the head
reduction fragment).

\subsection{$\LambdaLxr$}

In this section we provide the guidelines to obtain a similar completeness theorem for $\LambdaLxr$, leaving the details for further work.
The methodology is similar to the cases of the $\lambda$-calculus and $\LambdaS$:
we identify a notion of reduction for which Subject Expansion holds, and whose notion of normal forms can be easily typed.

As in the $\lambda$-calculus, some of the rules do not satisfy Subject Expansion:
\[
\begin{array}{lll}
  \Subst {\weakening x M} x N& \Rew{}& \weakening {fv(N)} M \\
  M {\weakening x N} &\Rew{}& \weakening x {M N} \\
  \Subst M x {\weakening y N} &\Rew{}& \weakening y {\Subst M x N} \\
  \contraction x y z {\weakening y M} &\Rew{}& \subst M z x 
\end{array}
\]
Subject Expansion for the other reduction rules should be straightforward.

The fact that the last 3 rules above do not satisfy Subject Expansion is not problematic for the completeness theorem: like rule $W$ in $\LambdaS$, we should prove that they can be postponed after the other rules, and that removing them from the system defines a new notion of normal forms that can still be typed.

On the other hand, the first rule above is more problematic: if we forbid it, the reduction can be blocked (like forbidding erasure can block a reduction in the pure $\lambda$-calculus). So a Subject Expansion result without that rule is not enough to prove the completeness of $\LambdaLxr$.
Hence, we have two possibilities to achieve that goal:
\begin{iteMize}{$\bullet$}
\item We adapt the proof of the pure $\l$-calculus. We have to define
  $\per_E$ and $\Per_E$ in $\LambdaLxr$.
\item We adapt the proof of $\LambdaS$. We cannot do this in the usual
  $\LambdaLxr$: if we forbid erasure, reduction can be blocked. So we have
  to add the rules that move $\Subst {\weakening x M} x N$
  (like $\Subst M x N$ with $x \notin fv(M)$ in $\LambdaS$).
  These added rules respect Subject Reduction, so we still have soundness.
\end{iteMize}













\noindent Both approaches would provide Completeness (strong normalisation implies
typability).
Moreover, as in the pure $\l$-calculus and $\LambdaS$,
the proof would provide an algorithm that constructs an optimal typing tree from a
strongly normlising term in $\LambdaLxr$.









\section{Complexity results}

\label{sec:complexity}

With the Subject Reduction theorems of the different calculi, we have proved that for every $\beta$- (or $B$-) reduction, the measure of the typing tree strictly decreases.  Hence, more than strong normalisation, this gives us a bound on the number of $\beta$- (or $B$-) reductions in a reduction sequence.  So we have a complexity result which is an inequality.  We would like to refine this result and have an equality instead.  The main idea is to only perform $\beta$- (or $B$-) reductions that decreases the measure of the typing tree by exactly one.  Given a term $M$ and any typing tree for it, it is not always possible to find such a reduction step. But it is always possible provided the typing tree is optimal.  Fortunately, every term $M$ that is typable is typable with an optimal typing tree: with soundness we can prove that $M$ is strongly normalising, and then, with completeness, we can prove that $M$ is typable with an optimal typing tree.  This is the main reason for introducing the notion of optimality. As in Section~\ref{sec:completeness}, the case of $\LambdaS$ is simpler than the case of pure $\l$-calculus, so we will deal with it first.

\subsection{$\LambdaS$}

In $\LambdaS$, we take advantage of the fact that $\Rew{W}$ can be postponed \wrt to $\Rew{B,S}$ steps.  This allows us to concentrate on $\Rew{B,S}$ and the normal forms for it.

\begin{lemma}[Refined Subject Reduction]
  If $\Deriliopt[n] \Gamma M A$ then:
  \begin{iteMize}{$\bullet$}
  \item If $M \Rew{B} M'$, then there exist $\Gamma'$ and $m$ such that $\Gamma \approx \Gamma'$, $m<n$ and $\Deriliopt[m] {\Gamma'} {M'} A$
  \item If $M \Rew{S} M'$, then there exists $\Gamma'$ such that $\Gamma \approx \Gamma'$ and $\Deriliopt[n] {\Gamma'} {M'} A$
  \end{iteMize}
  Moreover the degree of the typing tree does not change, and the principality property is preserved.
\end{lemma}
\begin{proof}
  We simply check that, in the proof of of Subject Reduction (Theorem~\ref{th:SubjectReduction-lS}), the optimality property, the degree and the principality property are preserved, as already mentioned in the proof of Subject Expansion (Theorem~\ref{th:SELambdaS}).
\end{proof}

\begin{toappendix}
  \appendixbeyond 0
  \begin{lemma}[Most inefficient reduction]
\label{lem:MostInefficientReduction}
    Assume $\Deriliopt[n] \Gamma M A$.
    If $M$ can be reduced by $\Rew{B}$ and not by $\Rew{S}$, then there exist $M'$ and $\Gamma'$ such that $\Gamma \approx \Gamma'$, $M \Rew{B} M'$ and $\Deriliopt[n - 1] {\Gamma'} {M'} A$.
  \end{lemma}
\end{toappendix}

\begin{toappendix}
  [\begin{proof} Again, we adapt the proof of Subject Reduction (Theorem~\ref{th:SubjectReduction-lS}). More precisely, see appendix \thisappendix.\end{proof}]
  \begin{proof}
    We follow the induction given in the proof of Subject Reduction (Theorem~\ref{th:SubjectReduction-lS}).
    In this induction, $n$ can be decreased by more than $1$ by a $\Rew{B}$ in two cases:
    \begin{iteMize}{$\bullet$}
    \item In the case where the type is an intersection, then $n$ will be decreased by at least $2$.
    \item When we build a typing of $\Subst M x N$ from a typing of $(\l x . M) N$: if there were subsumption in the typing the $\l$-abstraction, then the proof calls Lemma~\ref{lem:basicprop}.4 which might decrease $n$ by more than $1$.
    \end{iteMize}
    Those two cases are never encountered when optimality is assumed, as we prove the result by induction on $M$.
    Since $M$ cannot be reduced by $\Rew{S}$, it is of one of the following forms:
    \begin{iteMize}{$\bullet$}
    \item $\l x . M_1$. It is clear that $M_1$ satisfies the necessary conditions to apply the induction hypothesis.
    \item $(\l x . M_1) N_1 \ldots N_p$ (with $p\geq 1$). We reduce to $\Subst{M_1} x{N_1}\ N_2\ldots N_p$.
      By the optimality property, $A$ is not an intersection, and none of the types of $((\l x . M_1) N_1\ldots N_i)_{1\leq i\leq p-1}$ are intersections either (since they are applied to an argument).
      Also by the optimality property, there is no subsumption in the typing of the $\l$-abstraction, and therefore the call to Lemma~\ref{lem:basicprop}.4 is replaced by a call to Lemma~\ref{lem:basicprop}.4 and therefore $n$ is decresed by exactly $1$.
    \item $\Subst {\Subst x {y_1} {N_1} ...} {y_p} {N_p} N_{p + 1} ... M_m$.
      Therefore there exists $i$ such that $N_i$ can be reduced by $\Rew{B}$. Moreover, optimality requires the type of $x$ to be of the form $A^+_1\arr \cdots\arr A^+_p\arr B-$, and therefore the sub-derivation typing $N_i$ is also optimal: we can apply the induction hypothesis on it.\qedhere
    \end{iteMize}
  \end{proof}
\end{toappendix}

\begin{lemma}[Resources of a normal term]\strut
  \label{lem:normalLambdaS}
  If $\Deriliopt[n] \Gamma M A$, and $M$ cannot be reduced by $\Rew{B, S}$, then
  $n$ is the number of applications in $M$.

	Moreover, if the typing tree is principal, then $n$ is the degree of the
typing tree.
\end{lemma}

\begin{proof}
  Straightforward.
\end{proof}

\begin{theorem}[Complexity result]
  If $\Deriliopt[n] \Gamma M A$, then $n = n_1 + n_2$ where
  \begin{iteMize}{$\bullet$}
  \item $n_1$ is the maximum number of $\Rew{B}$ in a $B,S$-reduction sequence from $M$
  \item $n_2$ is the number of applications in the $B,S$ normal form of $M$.
  \end{iteMize}
Moreover, if the typing tree is principal, then $n_2$ is the degree of the
typing tree.
\end{theorem}

\begin{proof}
  The previous lemmas give us a $B,S$-reduction sequence with $n - n_2$ $B$-steps, from $M$ to the normal form of $M$.  In this reduction sequence every reduction $B$ decreases the measure of the typing tree by exactly one.

  Assume we have a $B,S$-reduction sequence, from $M$ to the normal form of $M$ ($\LambdaS$ is confluent), with $m$ $B$-steps. By Subject Reduction (Theorem~\ref{th:SubjectReduction-lS}), the measure of the derivation typing the normal form of $M$ is smaller than $n-m$, but is also $n_2$.  Hence $m \leq n-n_2$.
  
  Assume we have a $B,S$-reduction sequence, from $M$ to any term $M_1$. It can be completed into a $B,S$-reduction sequence with more $B$-steps, from $M$ to the normal form of $M_1$.
\end{proof}

\subsection{$\LambdaLxr$}

In this section we suppose that we have reductions to move
$\Subst {\weakening x M} x N$.

Therefore, it is reasonable to consider $\Rew{rev}$ which is all reductions
except the 5 ones that cause a problem for subject expansion.

Hence we can adapt the proofs of $\LambdaS$ and obtain the following theorem:

\begin{theorem}[Complexity result]If $\Deriliopt[n] \Gamma M A$, then $n = n_1 + n_2$ with $n_1$ the maximum number
of $B$ in a $\Rew{rev}$ sequence, and $n_2$ the number of applications in
the $\Rew{rev}$ normal form.

Moreover, if the typing tree is principal, then $n_2$ is the degree of the typing
tree.
\end{theorem}














\subsection{Pure $\l$-calculus}

The case of $\l$-calculus is harder because we cannot ignore erasure ($\beta$-reductions that erase sub-terms cannot always be postponed).
Therefore:
\begin{iteMize}{$\bullet$}
\item We will need to use the results we have for $\LambdaS$.
\item We will have to use degrees and principal typing trees.
  This is why we defined those two notions in the first place.
\end{iteMize}

\newcommand{\mysimplehead}{\Rightarrow_{head}}
\newcommand{\myhead}[2]{#1 \mysimplehead #2}

\noindent We produce and measure the longest $\beta$-reduction sequences by simply using the perpetual strategy from~\cite{RaamsdonkSSX1999}, shown in Fig.~\ref{fig:perpwonalc}.
\begin{figure}[!h]
  $$
  \begin{array}{|@{\quad}c@{\quad}|}
    \hline\\
    \prooftree{x\in \FV M\mbox{ or $M'$ is a $\beta$-normal form}}\justifies{(\l
      x.M)\; M'\; \v{M_j}\mysimplehead \subst M
      x{M'}\;\v{M_j}}
    \endprooftree
    \qquad
    \begin{prooftree}{M'\mysimplehead M''\qquad x\notin
        \FV M}\justifies{(\l x.M)\; M'\; \v{M_j}\mysimplehead (\l x.M)\; M''\;
        \v{M_j}}
    \end{prooftree}\\\\
    \prooftree{M\mysimplehead M'}\justifies{x\; \v{M_j}\;M\; \v{N_j}\mysimplehead x\;
      \v{M_j}\;M'\; \v{N_j}}
    \endprooftree \qquad
    \begin{prooftree}{M\mysimplehead M'}\justifies{\l x.M\mysimplehead \l x.M'}
    \end{prooftree}
    \\\\
    \hline
  \end{array}
  $$
  \caption{A perpetual reduction strategy for $\l$}
  \label{fig:perpwonalc}
\end{figure}

\begin{remark}Notice that this restricted reduction relation is a fragment of that defined in Fig.~\ref{fig:restrictbeta}:
  \[\mysimplehead \subseteq \Per_\emptyset \subseteq \Rew{\beta}\]
  Moreover, if $M$ is not a $\beta$-normal form, then there is a $\l$-term $M'$ such that $M\mysimplehead M'$.
\end{remark}

Although we do not need it here, it is worth mentioning that $\mysimplehead$ defines a perpetual strategy \wrt  $\beta$-reduction, \ie  if  $M$ is not $\beta$-strongly normalising and $M \mysimplehead M'$, then neither is $M'$~\cite{RaamsdonkSSX1999}. In that sense it can be seen as the worst strategy (the least efficient). We show here that it is the worst in a stronger sense: it maximises the lengths of reduction sequences.

\begin{lemma}[Resources of a normal term]\strut
\label{lem:ResourcesOfANormalTerm}

  If $\Deriliopt[n] \Gamma M A$ with a principal typing tree of degree $d$, and $M$ cannot be $\beta$-reduced, then $n = d$.

\end{lemma}
\begin{proof}
  If $M$ is a normal form for $\beta$, then $M$ is also a normal form for $B,S$, so we can apply Lemma~\ref{lem:normalLambdaS}.
\end{proof}

\begin{lemma}[Most inefficient reduction]\strut

  If $\Deriliopt[n] \Gamma M A$ a principal typing tree with a degree $d_1$ and $\myhead M {M'}$, then there exist $m$, $\Gamma$, $A'$ and $d_2$ such that $\Deriliopt[m] {\Gamma'} {M'} {A'}$ a principal typing tree with a degree $d_2$ and $n - d_1 = m - d_2 + 1$.
\end{lemma}

\begin{proof}
  By induction on $\myhead M {M'}$.
  The proof is an adaptation of Lemma~\ref{lem:MostInefficientReduction}.
  But, in Lemma~\ref{lem:MostInefficientReduction} we had $d_1 = d_2$,
  because there was no erasure.
  To control how the degree changes when there is erasure,
  we use Lemma~\ref{lem:ResourcesOfANormalTerm}.
  We can use it because $\mysimplehead$ only erases normal terms.
\end{proof}

\begin{lemma}[Relating $\l$ and $\LambdaS$]
 If $M \Rew[n]{\beta} M'$, then $M (\Rew{B}\Rewn{S})^{n} \Rewn{W} M'$.
\end{lemma}

\begin{proof}
  We proceed as follows:
  \begin{iteMize}{$\bullet$}
  \item Given two pure $\l$-terms $M$ and $N$, note that $\subst M x N$ is a pure $\l$-term, and it is easy to show, by induction on $M$ and using erasure postponement (Lemma~\ref{lem:erasurepostponed}), that $\Subst M x N \Rewn{S}\Rewn{W} \subst M x N$.
  \item Then we show, again by induction on $M$, that if $M \Rew{\beta} M'$, then $M \Rew{B}\Rewn{S}\Rewn{W} M'$.
  \end{iteMize}
The result is a direct corollary, obtained by induction on $n$ and using Lemma~\ref{lem:erasurepostponed}.
\end{proof}

\begin{theorem}[Complexity result]
  If $\Deriliopt[n] \Gamma M A$ with a principal typing tree of degree $d$, then the length of the longest $\beta$-reduction sequence from $M$ is $n - d$.
\end{theorem}

\begin{proof}
  Two previous lemmas give us a $\beta$-reduction sequence of size $n - d$.
  Let $L$ be another $\beta$-reduction sequence from $M$ of size $m$.
  So there exists $M_1$ such that $M \Rew[m]{\beta} M_1$.
  By the previous lemma, there exists $M_2$ such that
  $M (\Rew{B}\Rewn{S})^{m} M_2$ and $M_2 \Rewn{W} M_1$.
  From the complexity result for $\LambdaS$ we have $m \leq n - d$.
\end{proof}

Contrary to $\LambdaS$, we cannot have a complexity result on the weaker assumption of optimality, relating the measure to the number of applications in the normal form:
This was possible in $\LambdaS$ because we considered normal forms for a system that never erases, while here we cannot forbid $\beta$-reduction to erase terms.

\section{Other measures of complexity}

\label{sec:other}

In the pure $\l$-calculus we measure the (maximal) number of $\beta$-steps.  The equivalent result for $\LambdaS$ and $\LambdaLxr$ naturally counts the number of $\Rew{B}$-steps if we do not change the measure on the typing trees.  But there are many other reduction rules for $\LambdaS$ and $\LambdaLxr$, for which we may want similar complexity results. For some of these rules we can obtain such results without changing the typing system, by changing what we count in the typing trees.

\subsection{Number of replacements}

To get the number of replacements we measure in the typing tree the number of
use of the variable rule.

\begin{theorem}[Complexity result on the number of replacements]

The longest reduction sequence from $M$ by measuring the number of $\Rew{B}$
is the longest reduction sequence from $M$ by measuring $\Rew{SR}$
(head reduction strategy).

And the number of use of $\Rew{SR}$ in this sequence plus the number of
variables in the normal form (without weakening) is equal to the number of
use of the variable rule in an optimal typing tree.

\end{theorem}

\subsection{Number of duplications}

\label{sec:other-inter}

By measuring the number of use of the Intersection rule in the typing tree
we get a bound on the number of duplications (the number of use of rules that
duplicate a term).

However, contrary to the other measures, we cannot have an equality result.

Here is a counter example :

\[(\l x . x x)(\l y . a y y)\]

If we reduce this term to its normal form we have two duplications.
However, if we type this term, we have at least 3 uses of the intersection
rule.

\subsection{The other measures}

\null
\begin{iteMize}{$\bullet$}

\item If we measure the number of uses of the Abstraction rule we get a result
on the maximum number of $\Rew{B}$ in a reduction sequence again.
We just have to change the definition of degree of a principal typing tree.

\item The explicit substitution rule can be produced or destroyed by the subject
reduction.
So we cannot use it to get a complexity result.

\end{iteMize}

\section{Conclusion}
\label{sec:conclusion}

We have defined a typing system with non-idempotent intersection types. We have shown that it characterises strongly normalising terms, in the pure $\l$-calculus as well as in the explicit substitution calculi $\LambdaS$ and $\LambdaLxr$. This characterisation has been achieved in each case by strong versions of Subject Reduction and Subject Expansion, enriched with quantitative information:
\begin{iteMize}{$\bullet$}
\item By identifying a measure on typing derivations that is decreased by Subject Reduction, we have obtained a simple proof of strong normalisation that also provides upper bounds on longest reduction sequences.
\item By either proving postement of erasures ($\LambdaS$) or identifying appropriate sub-reduction relations ($\l$), we have shown how Subject Expansion garantees the existence of typing derivations satisfying extra properties (optimality and principality), where the bounds are refined into an exact measure of longest reduction sequences.
\end{iteMize}

\noindent In the case of $\lambda$-calculus, obtaining this exact equality departs from the issues addressed in \eg\cite{WellsPOPL99,NeergaardICFP04} whose technology is similar to ours (as we found out \emph{a posteriori}). Indeed, one of the concerns of this line of research is how the process of type inference compares to that of normalisation, in terms of complexity \emph{classes} (these two problems being parametrised by the size of terms and a notion of \emph{rank} for types). Here we have shown how such a technology can actually provide an exact equality specific to each $\lambda$-term and its typing tree. Of course this only emphasises the fact that type inference is as hard as normalisation, but type inference as a process is not a concern of this paper.

Moreover, we have extended those results to $\LambdaS$ and $\LambdaLxr$, and the technology can be adapted to other calculi featuring \eg\ combinators, or algebraic constructors and destructors (to handle integers, products, sums,\ldots).

We have seen how the use of non-idempotent intersection types simplifies the methodology from~\cite{CoquandSpiwack07} by cutting a second use of reducibility techniques to prove strong normalisation properties of standard systems (here illustrated by the examples of simple types, System $F$, and idempotent intersections). We extended the methodology to prove strong normalisation results for $\LambdaS$ and $\LambdaLxr$, providing the first direct proofs that we are aware of.

We have seen how the corresponding filter model construction can be done by orthogonality techniques; for this we have defined an abstract notion of orthogonality model which we have not seen formalised in the literature.  As illustrated in Section~\ref{sec:instances}, this notion allows a lot of work (\eg proving the Adequacy Lemma) to be factorised, while building models like \MSN, \Mcap\ and \Mfilters. Comparing such instances of orthogonality models, we have seen the superiority of \Mfilters\ for proving the strong normalisation results of $\LambdaS$ and $\LambdaLxr$. Note that, while \Mfilters\ and \Mfiltersi\ share the same ingredients $\ValDomE$, $\ValDom$, $\ap$ and $\SemTe \_\_$, they are different in the way types are interpreted; see the discussion in Appendix~\ref{sec:realisability}.

In~\cite{bernadetleng11b} we also compared the models in the way they enlighten the transformation of infinite polymorphism into finite polymorphism.
We leave this aspect for another paper, as more examples should be computed to illustrate (and better understand) the theoretical result; in particular we need to understand how and why the transformation of polymorphism does not require to reduce terms to their normal forms. An objective could be to identify (and eliminate), in the interpretation of a type from System~$F$, those filters that are not the interpretation of any term of that type. What could help this, is to force filters to be stable under type instantiation, in the view that interpretations of terms are generated by a single $F$-type, \ie a {\em principal} type.

Another aspect of this future work is to use the filter models to try to lift the complexity results that we have in the target system back into the source system, and see to what extent the quantitative information can be already read in the typing trees of the source system. One hope would be to recover for instance results that are known for the simply-typed calculus~\cite{Sch82,BeckmannJSL01}, but with our methodology that can be adapted to other source systems such as System~$F$.

Finally, the appropriate sub-reduction relation for the $\l$-calculus, which we have used to prove Subject Expansion as generally as possible, also helps understanding how and when the semantics $\SemTe\_\_$ of terms is preserved, see Appendix~\ref{sec:preserv}.  This is similar to~\cite{AlessiTCS06}, and future work should adapt their methodology to accommodate our non-idempotent intersections.

\para{Acknowledgement} The authors are grateful to the anonymous referees for their numerous constructive remarks (and for pointing out references).

\bibliographystyle{Common/good}
\bibliography{Common/abbrev-short,Common/Main,Common/crossrefs}

\appendix

\section{Filter models: classical vs. intuitionistic realisability}
\label{sec:realisability}

The orthogonality method comes from the denotational and operational semantics of symmetric calculi, such as proof-term calculi for classical or linear logic.

In some sense, orthogonality only builds semantics in a continuation passing style, and (as we have seen) this still makes sense for typed $\lambda$-calculi that are purely intuitionistic. While this is sufficient to prove important properties of typed $\lambda$-terms such as strong normalisation, the models are unable to reflect some of their purely intuitionistic features.

This phenomenon could be seen in presence of a ``positive'' type (\ie datatype) $\gP$, for which $\SemTyP \gP{}$ is {\bf not} closed under bi-orthogonal and $\SemTyN \gP{}$ is defined as $\uniorth{\SemTyP \gP{}}$.
Model \Mfilters\ provides the interpretation
\[\begin{array}{rl}
\SemTyP {\gP\arr \gP}{}&=\uniorth{(\cons{\SemTyP \gP{}}{\SemTyN \gP{}})}
=\{u\in \ValDom \mid \forall v \in \SemTyP \gP{}, \forall \vec {v'} \in \SemTyN \gP{},\orth u{\cons v{\vec {v'}}}\}\\
&=\{u\in \ValDom \mid \forall v \in \SemTyP \gP{}, \forall \vec {v'} \in \SemTyN \gP{},\orth {u\ap v}{{\vec {v'}}}\}\\
&=\{u\in \ValDom \mid \forall v \in \SemTyP \gP{}, u\ap v\in\biorth{\SemTyP \gP{}}\}\\
\end{array}
\]
while model \Mfiltersi\ would provide \[\SemTyP {\gP\arr \gP}{}=\{u\in \ValDom \mid \forall v \in \SemTyP \gP{}, u\ap v\in{\SemTyP \gP{}}\}\]


\section{Preservation of semantics by reduction}

\label{sec:preserv}

When models are built for a typed $\l$-calculus, it is sometimes expected that the interpretation of terms is preserved under $\beta$-reduction (or even $\beta$-equivalence).
It is not always necessary for the purpose of the model construction (here: proving normalisation properties), and it is clearly not the case for the term models \MSN and \Mcap, where terms are interpreted as themselves (at least in the case of the pure $\l$-calculus).
The case of the filter models \Mfilters\ and \Mfiltersi\ (which heavily rely on Theorem~\ref{th:SemTe}) is less obvious. Still we can prove the following:

\begin{theorem}\strut
\label{th:preserv}
  \begin{enumerate}[\em(1)]

  \item If $M \Rew{\beta} M'$, then for all $\rho$,
    $\SemTe M \rho \subseteq \SemTe {M'} \rho$.

	\item If $M \per_E M'$, then for all $\rho$,
		$\SemTe M \rho = \SemTe {M'} \rho$.

	\item If $M \Rew{B,S} M'$, then for all $\rho$,
		$\SemTe M \rho = \SemTe {M'} \rho$.




  \end{enumerate}

\end{theorem}

\begin{proof}\hfill
  \begin{enumerate}[(1)]
  \item Corollary of Subject Reduction (Theorem~\ref{th:SubjectReduction}).
  \item One inclusion is the previous case and the other is a corollary
  of Subject Expansion (Theorem~\ref{th:SubjectExpansion}).
  \item Same argument, with Subject reduction and Subject Expansion
		of $\LambdaS$ (Theorems~\ref{th:SubjectReduction-lS} and
		\ref{th:SELambdaS}).\qedhere
  \end{enumerate}
\end{proof}

\begin{example}There are cases where $\SemTe M \rho \neq \SemTe {M'} \rho$:\\
an obvious example is when $M\notin\SNLambda{}$ but $M'\in\SNLambda{}$ (\eg $M \eqdef (\l y . z) ((\l x . x x) (\l x . x x))$ and $M' \eqdef z$); there exists $\rho$ such that $\SemTe M \rho = \bot$ and $\SemTe {M'} \rho \neq \bot$.
\end{example}

We can also find an example where $M \Rew{\beta} M'$, $\SemTe M \rho \neq \bot$, and $\SemTe M \rho \neq \SemTe {M'} \rho$:
\begin{example}
\label{ex:last} Take $M \eqdef (\l z . (\l y . a) (z z))$ and $M' \eqdef \l z . a$.
\end{example}
\begin{proof}
  Suppose that $\SemTe M \rho = \SemTe {M'} \rho$ with
  $\rho = (a \mapsto \top)$. Then:
  \begin{iteMize}{$\bullet$}
  \item $N = \l x . x x\in \SNLambda{}$ and closed, so $\SemTe N \rho \neq \bot$.
  \item $\SemTe M \rho \ap \SemTe N \rho =
   \SemTe {(\l y . a) (z z)} {\rho, z \mapsto \SemTe N \rho} =
    \SemTe {(\l y . a)} {\rho, z \mapsto \SemTe N \rho} \ap \SemTe {z z}
    {\rho, z \mapsto \SemTe N \rho}\\ =
    \SemTe {\l y .a} \rho \ap \SemTe {N N} {\rho} =
    \SemTe{\l y . a} \rho \ap \bot = \bot$
  \item $\SemTe {M'} \rho \ap \SemTe N \rho =
    \SemTe a {\rho, z \mapsto \SemTe N \rho} = \top$
  \end{iteMize}
Hence $\top = \bot$, which is a contradiction.
\end{proof}

Notice that this proof only uses the properties expected from the
model (Theorem~\ref{th:SemTe} and the characterisation of $\SNLambda{}$) and not the construction of the model itself.

\section{Full proofs}

\gettoappendix {lem:capprop}
\gettoappendix {lem:cappropproof}
\gettoappendix {lemma:TypingSubst}
\gettoappendix {lemma:TypingSubstproof}
\gettoappendix {lem:typsubSR}
\gettoappendix {lem:typsubSRproof}
\gettoappendix {th:SubjectReduction}
\gettoappendix {th:SubjectReductionproof}
\gettoappendix {lem:SNForSW}
\gettoappendix {lem:SNForSWproof}
\gettoappendix {th:SubjectReduction-lS}
\gettoappendix {th:SubjectReduction-lSproof}
\gettoappendix {fig:LambdaLxr_reductions}
\gettoappendix {fig:LambdaLxr_reductionsproof}
\gettoappendix {th:SemTe}
\gettoappendix {th:SemTeproof}
\gettoappendix {lem:Adequacy}
\gettoappendix {lem:Adequacyproof}
\gettoappendix {lem:MostInefficientReduction}
\gettoappendix {lem:MostInefficientReductionproof}

\end{document}

%% file: Common/Macros_Main.tex
\long\def\ignore#1{\relax}

\newcommand\struto[1][15pt]{{\raise #1 \hbox{\strut}}}%
\newcommand\strutb[1][15pt]{{\raise-#1 \hbox{\strut}}}%

\newcommand\upline{\hline\struto[12pt]}

\newcommand\downline[1][12]{\\[+ #1pt]\hline}

 \newcommand\toaux[1]{\immediate\write\@auxout{#1}}

\newcommand\mybox[1]{\fbox{\vbox{#1}}}
\renewcommand\[[1][3]{\par\removelastskip\vskip#1pt\vbox\bgroup\hrule height0pt\vfil\hbox to\hsize\bgroup\hfil\(}
\renewcommand\][1][3]{\)\hfil\egroup\vfil\hrule height0pt\egroup\vskip#1pt\nointerlineskip\noindent}

\newbox\columnsbox
\newbox\tmpbox
\newdimen\columnsheight
\newdimen\columnwidth
\newdimen\remainingwidth
\newdimen\textwidthsave
\def\mycolumnsheight{}

\newcommand\columns[1]{%
  \def\mycolumnsheight{}%
  \setlength\remainingwidth\textwidth%
  \setbox\columnsbox=\vbox\bgroup\vskip0pt\vfil\hbox to\textwidth\bgroup#1\egroup\vfil\egroup%
  \columnsheight=\ht\columnsbox%
  \def\mycolumnsheight{to\columnsheight}%
  \hrule height 0pt\vtop{\hbox to\wd\columnsbox\bgroup#1\egroup}%
}

\makeatletter

\def\commonpart{%
  \setlength\columnwidth{\wd\tmpbox}%
  \vtop{\vskip0pt\hbox to\columnwidth{{\box\tmpbox}}}%
  \advance\remainingwidth-\columnwidth%
  \setlength\textwidth\textwidthsave%
  \hsize\textwidthsave%
}
\def\column{\unskip\setlength\textwidthsave\textwidth\@ifnextchar[\@columnwith\@columnwithout}
\long\def\@columnwith[#1]#2{%
  \def\newhsize{#1\dimexpr\textwidth\relax}%
  \hsize\newhsize%
  \ifdim\hsize<0.1pt\hsize\remainingwidth\fi%
  \setlength\textwidth\hsize%
  \setbox\tmpbox=\hbox to\hsize\bgroup\hfil\vtop\mycolumnsheight{\vskip0pt#2\vskip0pt}\hfil\egroup%
  \commonpart%
}
\long\def\@columnwithout#1{%
  \hsize\remainingwidth%
  \setlength\textwidth\hsize%
  \setbox\tmpbox=\hbox\bgroup\vtop\mycolumnsheight{\vskip0pt#1\vskip0pt}\egroup%
  \commonpart%
}
\makeatother

\renewcommand\v[1]{{\overrightarrow{#1}}}            

\newcommand\dom[1]{\textsf{Dom}{(#1)}}

\newcommand{\eqdef}{:=\ }
\newcommand{\recdef}{::=\ }

\renewcommand\l{\lambda}

\newcommand{\Gam}{\Gamma}

\newcommand{\Del}{\Delta}

\newcommand\gGamma{\mathfrak{G}}
\newcommand\gDelta{\mathfrak{H}}
\newcommand\gA{\mathfrak{A}}
\newcommand\gB{\mathfrak{B}}

\newcommand\gP{\mathfrak{P}}

\newcommand\lx{$\l${\textsf x}}

\newcommand\mathFomega{F_\omega}
\newcommand\Fomega{\ifmmode\mathFomega\else$\mathFomega$\fi}
\newcommand\mathFomegaC{F_\omega^{\mathcal C}}
\newcommand\FomegaC{\ifmmode\mathFomegaC\else$\mathFomegaC$\fi}
\newcommand\mathDNE{\mathrm{DNE}}
\newcommand\DNE{\ifmmode\mathDNE\else$\mathDNE$\fi}

\renewcommand{\iff}{if and only if}

\newcommand{\ie}{i.e.~}
\newcommand{\eg}{e.g.~}
\newcommand{\cf}{cf.~}
\newcommand{\etc}{etc.~}
\newcommand{\wrt}{w.r.t.~}
\newcommand{\resp}{resp.~}

%% file: Common/Macros_Bib.tex
\ifdefined\url
\else
\def\url[#1]#2{\texttt{#2}}
\fi

\let\oldurl\url

\makeatletter
\ifdefined\href
\def\myurl{\@ifnextchar[\@myurlwith\@myurlwithout}
\long\def\@myurlwith[#1]#2{\mbox{\href{#2}{#1}}}
\long\def\@myurlwithout#1{\mbox{\href{#1}{#1}}}
\else
\def\myurl{\@ifnextchar[\@myurlwith\@myurlwithout}
\long\def\@myurlwith[#1]#2{\mbox{\oldurl[#1]{#2}}}
\long\def\@myurlwithout#1{\mbox{\oldurl{#1}}}
\fi
\makeatother

\def\url{\myurl}

\newcommand\monthdisplay[1]{\unskip}

%% file: Macros.tex
\newcommand{\subst}[3]{ #1 \{#2 := #3 \} }
\newcommand{\Subst}[3]{ #1 [#2 := #3]}

\newcommand{\weakening}[2]{W_{#1}({#2})}
\newcommand{\contraction}[4]{C_{#1}^{{#2}, {#3}}({#4})}

\newcommand{\LambdaS}{\l S}
\newcommand{\LambdaLxr}{\l lxr}

\newcommand\TP[1]{\textsf{TP}(#1)}

\newcommand\Mfilters{\mbox{$\mathcal M_{\mathcal F}^{\bot\hspace{-7.1pt}\bot}$}}
\newcommand\Mfiltersi{$\mathcal M_{\mathcal F}^{\textsf i}$}

\newcommand\cons[2]{{#1}\!::\!{#2}}
\newcommand\el{[]}

\newcommand\SysS{\mathcal S}

\newcommand\TV[1]{ftv(#1)}

\newcommand{\Gen}[1]{< #1 >}
\newcommand{\SemTe}[3][]{{\llbracket #2 \rrbracket}^{#1}_{#3}}
\newcommand{\SemTyP}[2]{{\llbracket #1 \rrbracket}_{#2}}
\newcommand{\SemTyN}[2]{{[ #1 ]}_{#2}}
\newcommand\termmodel{\textsc{term}}
\newcommand\aptermmodel{\ap^{\textsc{term}}}

\newcommand\ValDom{\mathcal D}
\newcommand\ValDomL{\mathcal D^*}
\newcommand\ValDomE{\mathcal E}
\newcommand\orth[2]{{#1}\ \mbox{$\bot\!\!\!\!\bot$}\ {#2}}
\newcommand\uniorth[1]{#1^\bot}
\newcommand\biorth[1]{#1^{\bot\bot}}
\newcommand\triorth[1]{#1^{\bot\bot\bot}}
\newcommand\ap{@}

\newcommand\bottom{\bot}

\newcommand{\Rew}[2][]{\longrightarrow^{#1}_{#2}\;}
\newcommand{\Rewn}[1]{\longrightarrow^{*}_{#1}\;}

\newcommand{\tempty}[0]{\omega}

\newcommand\sep\mid

\def\SN#1{\textsf{SN}^{#1}}

\newcommand\col{\!:\!}

\newcommand\arr{\hspace{-3pt}\rightarrow\hspace{-3pt}}

\newcommand{\seqf}[2][]{\mbox{$\ {\pmb \vdash}_{#1}^{#2}\ $}}

\newcommand\Seqf[4][]{#3\seqf[{#1}]{#2} #4}

\newcommand\Derif[5][]{\Seqf[#1]{#2}{#3}{{#4}\col{#5}}}

\newcommand\Deri[4][]{\Derif[\cap\subseteq]{#1}{#2}{#3}{#4}}

\newcommand\Derilis[4][]{\Derif[\cap\subseteq]{#1}{#2}{#3}{#4}}
\newcommand\Derili[4][]{\Derif[\cap\subseteq]{#1}{#2}{#3}{#4}}
\newcommand\DeriliS[4][]{\Derif[\mathcal S]{#1}{#2}{#3}{#4}}
\newcommand\Deriliopt[4][]{\Derif[{^{\sf opt}}]{#1}{#2}{#3}{#4}}

\newcommand\InterLambdaSub{$\lambda_\cap^\subseteq$}

\newcommand\per{\rightsquigarrow}
\newcommand\Per{\Rightarrow}

\newcommand\FV[1]{fv(#1)}

\newcommand\degr[2][]{\delta^{#1}(#2)}
\newcommand\forg[1]{\textsf{forg}(#1)}

\newcommand\accu[2]{Acc({#1}, {#2})}

\newcommand\MSN{\mbox{$\mathcal M_{\SN{}}^{\bot\hspace{-7.1pt}\bot}$}}
\newcommand\Mcap{\mbox{$\mathcal M_{\cap}^{\bot\hspace{-7.1pt}\bot}$}}

\newcommand{\SNLambda}{\SN{\lambda}}
\newcommand{\SNLambdaS}{\SN{\LambdaS}}
\newcommand{\SNLambdaLxr}{\SN{\LambdaLxr}}

\newcommand\LambdaTerms[1][]{\Lambda^\lambda}
\newcommand\LambdaSTerms[1][]{\Lambda^{\LambdaS}}
\newcommand\LambdaLxrTerms[1][]{\Lambda^{\LambdaLxr}}

\newcommand\TypableTerms{\Lambda^\lambda_\cap}